\documentclass[preprint,11pt]{elsarticle}

\usepackage[margin=0.62in]{geometry}
\usepackage{bm}
\usepackage{framed}

\usepackage{xfrac}
\usepackage{amsthm}
\usepackage{amsmath}
\usepackage{algorithm}
\usepackage[noend]{algpseudocode}
\usepackage{amsfonts}

\newtheorem{mylemma}[]{Lemma}
\newtheorem{theorem}{Theorem}
\newtheorem{lemma}{Lemma}
\newtheorem{corollary}{Corollary}
\theoremstyle{definition}
\newtheorem{definition}{Definition}

\newenvironment{claim}[1]{\par\noindent\underline{Claim:}\space#1}{}

\makeatletter
\def\BState{\State\hskip-\ALG@thistlm}
\makeatother

\usepackage{algpseudocode}
\usepackage{algorithm}

\usepackage{amssymb}

\journal{}

\begin{document}

\begin{frontmatter}



\title{A Fast Distributed Solver for Symmetric Diagonally Dominant Linear Equations} 


\author{Rasul Tutunov$^{\dagger}$, Haitham Bou Ammar$^{\dagger}$, Ali Jadbabaie$^{\nmid}$}

\address{\{tutunov, haithamb, jadbabai\} @seas.upenn.edu\\ 
$\dagger:$ University of Pennsylvania, Computer and Information Science Dept.\\
Levine Hall 
3330 Walnut Street 
Philadelphia, PA 19104-6309 \\
$\nmid:$ University of Pennsylvania, Dept. of Electrical and System Engineering\\ Levine Hall 
3330 Walnut Street 
Philadelphia, PA 19104-6309 \\
}

\begin{abstract}
In this paper, we propose a fast distributed solver for linear equations given by symmetric diagonally dominant M-Matrices.  Our approach is based on a distributed implementation of the parallel solver of Spielman and Peng by considering a specific approximated inverse chain which can be computed efficiently in  a distributed fashion. 
Representing the system of equations by a graph $\mathbb{G}$, the proposed distributed algorithm is capable of attaining $\epsilon$-close solutions (for arbitrary $\epsilon$) in time propotional to $n^{3}$ (number of nodes in $\mathbb{G}$), $\bm{\alpha}$ (upper bound on the size of the R-Hop neighborhood), and $\frac{\bm{W}_{\text{max}}}{\bm{W}_{\text{min}}}$ (maximum and minimum weights of edges in $\mathbb{G}$). 
\end{abstract}



\end{frontmatter}


\newpage


\section{Introduction}
Solving systems of linear equations in symmetric diagonally matrices (SDD) is of interest to researchers in a variety of fields including but not limited to, solutions to partial differential equations~\cite{ExamplesOne}, computations of maximum flows in graphs~\cite{ExamplesTwo}, machine learning~\cite{Zhu03semi-supervisedlearning}, and   as basis for various algorithms~\cite{ExamplesThree}. 

Much interest has been devoted to determining fast algorithms for solving SDD systems. Spielman and Teng~\cite{DaanS} proposed a nearly linear-time algorithm for solving SDD systems, which benefited from the multi-level framework of~\cite{Reif199837,joshi}, preconditioners~\cite{citeulike}, and spectral graph sparsifiers~\cite{Batson,Spiel}. Further exploiting these ingredients, Koutis \textit{et. al}~\cite{Koutis, Koutis2} developed an even faster algorithm for acquiring $\epsilon$-close solutions to SDD linear systems. Further improvements have been discovered by Kelner \textit{et. al}~\cite{Kelner}, where their algorithm relied on only spanning-trees and eliminated the need for graph sparsifiers and the multi-level framework.  

On the parallel side, much progress has been made on developing  such solvers. Koutis and Miller~\cite{KoutisMiller07} proposed an algorithm requiring nearly-linear work and $m^{\sfrac{1}{6}}$ depth for planar graphs. This was then extended to general graphs by~\cite{Blelloch} leading to depth close to $m^{\sfrac{1}{3}}$. Peng and Spielman~\cite{Spielman} have proposed an efficient parallel solver requiring nearly-linear work and poly-logarithmic depth without the need for low-stretch spanning trees. Their algorithm, which we distribute in this paper, requires sparse approximate inverse chains~\cite{Spielman} which facilitates the solution of the SDD system. 

Less progress, on the other hand, has been made on the distributed version of these solvers. Current methods, e.g., Jacobi iteration~\cite{Axelsson,Bertsekas:1989:PDC:59912}, can be used for distributed solutions but require substantial complexity. In~\cite{liu2013asynchronous}, the authors propose a gossiping framework which can be used for a distributed solution of the above linear system. 
Recent work~\cite{Asuman} considers a local and asynchronous solution for solving systems of linear equations, where they acquire a bound on the number of needed multiplication proportional to the  degree and condition number for one component of the solution vector.

\textbf{Contributions:} In this paper, we propose a fast distributed solver for linear equations given by symmetric diagonally dominant M-Matrices. Our approach distributes the parallel solver in~\cite{Spielman} by considering a specific approximated inverse chain which can be computed efficiently in a distributed fashion. Our algorithm's computational complexity is given by  $\mathcal{O}\left(n^{3}\frac{\bm{\alpha}}{R}\frac{\bm{W}_{\text{max}}}{\bm{W}_{\text{min}}}\log\left(\frac{1}{\epsilon}\right)\right)$, with $n$ being the number of nodes in graph $\mathbb{G}$, $\bm{W}_{\text{max}}$ and $\bm{W}_{\text{min}}$ denoting the largest and smaller weights of the edges in $\mathbb{G}$, respectively, $\bm{\alpha}=\min\left\{n,\frac{d_{\text{max}}^{R+1}-1}{d_{\text{max}}-1}\right\}$ representing the upper bound on the size of the R-Hop neighborhood $\forall \bm{v} \in \mathbb{V}$, and $\epsilon \in (0,\frac{1}{2}]$ being the precision parameter. Our approach improves current linear methods by a factor of $\log n$ and by a factor of the degree  compared to~ \cite{Asuman} for each component of the solution vector.



\section{Problem Definition \& Notation}\label{Sec:Pre}
We consider the following system of linear equations:
\begin{equation}\label{lin_sys}
\bm{M}_{0}\bm{x} = \bm{b}_{0}
\end{equation}
where $\bm{M}_{0}$ is a Symmetric Diagonally Dominant M-Matrix (SDDM). Namely, $\bm{M}_{0}$ is symmetric positive definite with non-positive off diagonal elements, such that for all $i=1,2,\ldots, n$:
\begin{equation*}
\left[\bm{M}_{0}\right]_{ii} \ge -\sum_{j=1, j\ne i}^{n}\left[\bm{M}_{0}\right]_{ij}
\end{equation*}
The system of Equations in~\ref{lin_sys} can be interpreted as representing an undirected weighted graph, $\mathbb{G}$, with $\bm{M}_{0}$ being its Laplacian. Namely, $\mathbb{G} = \left(\mathbb{V},\mathbb{E},\bm{W}\right)$, with $\mathbb{V}$ representing the set of nodes, $\mathbb{E}$ denoting the edges, and $\bm{W}$ representing the weighted graph adjacency. Nodes $\bm{v}_i$ and $\bm{v}_j$ are connected with an edge $\bm{e}=\left(i,j\right)$ iff $\bm{W}_{ij}> 0$, where: 
\begin{equation*}
\bm{W}_{ij} = \left[\bm{M}_{0}\right]_{ii} \ \ \ \text{(if $i = j$)}, \ \ \ \ \text{or} \ \ \ \  \bm{W}_{ij} = -\left[\bm{M}_{0}\right]_{ij}, \ \ \ \text{otherwise}. 
\end{equation*}
Following~\cite{Spielman}, we seek $\epsilon$-approximate solutions to $\bm{x}^{\star}$, being the exact solution of $\bm{M}_{0}\bm{x}=\bm{b}_{0}$, defined as:
\begin{definition}
Let $\bm{x}^{\star}\in \mathbb{R}^{n}$ be the solution of $\bm{M}\bm{x}=\bm{b}_{0}$. A vector $\tilde{\bm{x}}\in \mathbb{R}^{n}$ is called an $\epsilon-$ approximate solution, if:
\begin{equation}
\left|\left|\bm{x}^{\star} - \tilde{\bm{x}}\right|\right|_{\bm{M}_{0}} \le \epsilon\left|\left|\bm{x}^{\star}\right|\right|_{\bm{M}_{0}}, \ \ \ \text{where $\left|\left|\bm{u}\right|\right|^{2}_{\bm{M}_0} = \bm{u}^{\mathsf{T}}\bm{M}_{0}\bm{u}$.
}
\end{equation}
\end{definition}

The R-hop neighbourhood of node $\bm{v}_{k}$ is defined as $\mathbb{N}_{r}\left(\bm{v}_{k}\right) = \{\bm{v}\in \mathbb{V}: \text{dist}\left(\bm{v}_{k}, \bm{v}\right)\le r\}$. We also make use of the diameter of a graph, $\mathbb{G}$, defined as $\text{diam}\left(\mathbb{G}\right) = \max_{\bm{v}_{i},\bm{v}_{j}\in \mathbb{V}}\text{dist}\left(\bm{v}_i,\bm{v}_j\right)$. 

\begin{definition}
We say that a matrix $\bm{A} \in \mathbb{R}^{n\times n}$ has a sparsity pattern corresponding to the R-hop neighborhood if $\bm{A}_{ij} = 0$ for all $i = 1,\ldots, n$ and for all $j$ such that $\bm{v}_j\notin \mathbb{N}_{r}\left(\bm{v}_i\right)$. 
\end{definition}

We will denote the spectral radius of a matrix $\bm{A}$ by $\rho\left(\bm{A}\right) = \max{\left|\bm{\lambda}_i\right|}$, where $\bm{\lambda}_i$ represents an eigenvalue of the matrix $\bm{A}$. Furthermore, we will make use of the condition number\footnote{Please note that in the case of the graph Laplacian,  the condition number is defined as the ratio of the largest to the smallest nonzero eigenvalues.}, $\kappa\left(\bm{A}\right)$ of a matrix $\bm{A}$ defined as $\kappa=\left|\frac{\bm{\lambda}_{\text{max}}\left(\bm{A}\right)}{\bm{\lambda}_{\text{min}}\left(\bm{A}\right)}\right|$.  In~\cite{DaanS} it is shown that the condition number of the graph Laplacian is at most $\mathcal{O}\left(n^{3}\frac{\bm{W}_{\text{max}}}{\bm{W}_{\text{min}}}\right)$, where $\bm{W}_{\text{max}}$ and $\bm{W}_{\text{min}}$ represent the largest and the smallest edge weights in $\mathbb{G}$. Finally, the condition number of a sub-matrix of the Laplacian is at most $\mathcal{O}\left(n^4\frac{\bm{W}_{\text{max}}}{\bm{W}_{\text{min}}}\right)$, see~\cite{Spielman}.

\subsection{Problem Definition}\label{Sec:ProbDef}
We assume that each node $\bm{v}_k\in \mathbb{V}$ has information about the weights of adjacent edges. Further, each node $\bm{v}_{k}$ has the capabilities of storing the value of the $i^{th}$ component of $\bm{b}_{0}$, which is denoted as $\left[\bm{b}_{0}\right]_i$. 
At each time step, nodes can exchange information with their neighboors. Each node is responsible for determining the corresponding component, $\bm{x}_{i}$, of the solution vector $\bm{x} \in \mathbb{R}^{n}$. We also assume a synchronized model whereby time complexity is measured by a global clock. The \textbf{goal} is to find $\epsilon$-approximate solution for $\bm{M}_{0}\bm{x}=\bm{b}_{0}$ in a \emph{distributed fashion}, while being restricted to R-hop communication between the nodes. 

\section{Background}
\subsection{Standard Splittings \& Approximations}\label{Sec:Standard}
Following the setup in~\cite{Spielman}, we provide standard definitions required in the remainder of the paper:
\begin{definition}
The standard splitting of a symmetric matrix $\bm{M}_{0}$ is: 
\begin{equation}
\bm{M}_{0} = \bm{D}_{0} - \bm{A}_{0}
\end{equation}
here $\bm{D}_0$ is a diagonal matrix such that $\left[\bm{D}_{0}\right]_{ii} = \left[\bm{M}_{0}\right]_{ii}$ for $i = 1,2,\ldots,n$, and $\bm{A}_0$ representing a non-negative symmetric matrix such that $\left[\bm{A}_0\right]_{ij} = -\left[\bm{M}_{0}\right]_{ij}$ if $i\ne j$, and $\left[\bm{A}_{0}\right]_{ii} = 0$.
\end{definition}
We also define the Loewner ordering: 
\begin{definition}
Let $\mathcal{\bm{S}}_{(n)}$ be the space of $n \times n$-symmetric matrices. The Loewner ordering $\preceq$ is a partial order on $\mathcal{\bm{S}}_{(n)}$ such that $\bm{Y}\preceq \bm{X}$ if and only if $\bm{X} - \bm{Y}$ is positive semidefinite.
\end{definition}

Finally, we define the ``$\approx_{\alpha}$'' operation used in the sequel to come as:
\begin{definition}
Let $\bm{X}$ and $\bm{Y}$ be positive semidefinite symmetric matrices. Then $\bm{X}\approx_{\alpha} \bm{Y}$ if and only iff
\begin{equation}
e^{-\alpha}\bm{X} \preceq \bm{Y} \preceq e^{\alpha}\bm{X}
\end{equation}
with $\bm{A}\preceq \bm{B}$ meaning $\bm{B} - \bm{A}$ is positive semidefinite.
\end{definition}

Based on the above definitions, the following lemma represents the basic characteristics of the $\approx_{\alpha}$ operator:
\begin{lemma}~\cite{Spielman}\label{approx_lemma_facts}
Let $\bm{X},\bm{Y},\bm{Z}$ and, $\bm{Q}$ be symmetric positive semi definite matrices. Then
\begin{enumerate}
\item[] (1) If $\bm{X}\approx_{\alpha} \bm{Y}$, then $\bm{X} + \bm{Z} \approx_{\alpha} \bm{Y} + \bm{Z}$, (2) If $\bm{X}\approx_{\alpha} \bm{Y}$ and $\bm{Z}\approx_{\alpha} \bm{Q}$, then $\bm{X} + \bm{Z} \approx_{\alpha} \bm{Y} + \bm{Q}$
\item[] (3) If $\bm{X}\approx_{\alpha} \bm{Y}$ and $\bm{Z}\approx_{\alpha} \bm{Q}$, then $\bm{X} + \bm{Z} \approx_{\alpha} \bm{Y} + \bm{Q}$, (4) If $\bm{X}\approx_{\alpha_1} \bm{Y}$ and $\bm{Y} \approx_{\alpha_2} \bm{Z}$, then $\bm{X} \approx_{\alpha_1 + \alpha_2} \bm{Z}$
\item[] (5) If $\bm{X}$, and $\bm{Y}$ are non singular and $\bm{X}\approx_{\alpha} \bm{Y}$, then $\bm{X}^{-1}\approx_{\alpha} \bm{Y}^{-1}$, (6) If $\bm{X}\approx_{\alpha} \bm{Y}$ and $\bm{V}$ is a matrix, then $\bm{V}^{\mathsf{T}}\bm{X}\bm{V}\approx_{\alpha}\bm{V}^{\mathsf{T}}\bm{Y}\bm{V}$
\end{enumerate}
\end{lemma}

The next lemma shows that good approximations of $\bm{M}^{-1}_0$ guarantee good approximated solutions of $\bm{M}_{0}\bm{x}=\bm{b}_{0}$.


\begin{lemma}\label{lemma_approx_matrix_inverse}
Let $\bm{Z}_0\approx_{\epsilon}\bm{M}^{-1}_0$, and $\tilde{\bm{x}} = \bm{Z}_0\bm{b}_0$. Then $\tilde{\bm{x}}$ is $\sqrt{2^{\epsilon}(e^{\epsilon} - 1)}$ approximate solution of $\bm{M}_{0}\bm{x}=\bm{b}_{0}$.
\end{lemma}
\begin{proof}
The proof can be found in the appendix. 
\end{proof}

We next discuss the parallel SDDM solver introduced in~\cite{Spielman}. 
\subsection{The Parallel SDDM Solver}\label{Sec:ParrallelSolver}
The parallel SDDM solver proposed in~\cite{Spielman} is a parallelized technique for solving the problem of Section~\ref{Sec:ProbDef}. It makes use of inverse approximated chains (see Definition~\ref{Def:InvChain}) to determine $\tilde{\bm{x}}$ and can be split in two steps. In the first step, denoted as Algorithm~\ref{Algo:Inv}, a ``crude'' approximation, $\bm{x}_{0}$, of $\bm{\tilde{x}}$ is returned. $\bm{x}_{0}$ is driven to the $\epsilon$-close solution, $\tilde{\bm{x}}$, using Richardson Preconditioning in Algorithm~\ref{Algo:Inv2}. Before we proceed, we start with the following two Lemmas which enable the definition of inverse chain approximation. 

\begin{lemma}~\cite{Spielman}\label{SDDM_splitting_lemma}
If $\bm{M} = \bm{D} - \bm{A}$ is an SDDM matrix, with $\bm{D}$ being positive diagonal, and $\bm{A}$ denoting a non-negative symmetric matrix, then $\bm{D} - \bm{A}\bm{D}^{-1}\bm{A}$ is also SDDM.
\end{lemma}

\begin{lemma}~\cite{Spielman}\label{approx_inverse_formulae_lemma}
Let $\bm{M} =\bm{D} - \bm{A}$ be an SDDM matrix, where $\bm{D}$ is positive diagonal and, $\bm{A}$ a symmetric matrix. Then 
\begin{align}\label{Inv_of_SDDM}
\left(\bm{D}-\bm{A}\right)^{-1} = \frac{1}{2}\Big[\bm{D}^{-1} + \left(\bm{I} + \bm{D}^{-1}\bm{A}\right)\left(\bm{D} - \bm{A}\bm{D}^{-1}\bm{A}\right)^{-1} \left(\bm{I} + \bm{A}\bm{D}^{-1}\right)\Big]
\end{align}
\end{lemma}


Given the results in Lemmas \ref{SDDM_splitting_lemma} and \ref{approx_inverse_formulae_lemma}, we now can consider inverse approximated chains of $\bm{M}_0$:
\begin{definition}\label{Def:InvChain}
Let $\mathcal{C} = \{\bm{M}_0, \bm{M}_1, \ldots, \bm{M}_d\}$  be a collection of SDDM matrices such that $\bm{M}_i = \bm{D}_i - \bm{A}_i$, with $\bm{D}_i$ a positive diagonal matrix, and $\bm{A}_i$ denoting a non-negative symmetric matrix. Then $\mathcal{C}$ is an inverse approximated chain if there exists  positive real numbers $\epsilon_0, \epsilon_1, \ldots, \epsilon_d$ such that: (1) For $i = 1,\ldots, d$: $\bm{D}_i - \bm{A}_i \approx_{e_{i-1}} \bm{D}_{i-1} - \bm{A}_{i-1}\bm{D}^{-1}_{i-1}\bm{A}_{i-1}$, (2) For $i = 1,\ldots, d$: $\bm{D}_i \approx_{\epsilon_{i-1}}\bm{D}_{i-1}$, and (3) $\bm{D}_d\approx_{\epsilon_d} \bm{D}_d - \bm{A}_d$. 
\end{definition}
 

\begin{algorithm}
  \caption{$\text{ParallelRSolve}\left(\bm{M}_0,\bm{M}_1,\ldots, \bm{M}_d, \bm{b}_0\right)$}
  \label{Algo:Inv}
  \begin{algorithmic}[1]
	\State \textbf{Input}: Inverse approximated chain, $\{\bm{M}_0,\bm{M}_1,\ldots, \bm{M}_d\}$, and $\bm{b}_0$ being 
	
	\State \textbf{Output}: The ``crude'' approximation, $\bm{x}_0$, of $\bm{x}^{\star}$    
    \For  {$i=1$ to $d$} 
    	\State $\bm{b}_i = \left(\bm{I}+\bm{A}_{i-1}\bm{D}^{-1}_{i-1}\right)\bm{b}_{i-1}$
    \EndFor \textbf{end for} 
    \State $\bm{x}_d = \bm{D}^{-1}_d\bm{b}_d$
    \For  {$i=d-1$ to $0$} 
    	\State $\bm{x}_{i} = \frac{1}{2}\left[\bm{D}^{-1}_i\bm{b}_{i} + \left(\bm{I} + \bm{D}^{-1}_i\bm{A}_i\right)\bm{x}_{i+1}\right]$
    \EndFor \textbf{end for}
    \State \textbf{return} $x_0$ 
  \end{algorithmic}
\end{algorithm}

The quality of the ``crude'' solution returned by Algorithm~\ref{Algo:Inv} is quantified in the following lemma:
\begin{lemma}~\cite{Spielman}\label{Rude_Alg_guarantee_Lemma}
Let $\{\bm{M}_0, \bm{M}_1,\ldots, \bm{M}_d\}$ be the inverse approximated chain and denote $\bm{Z}_0$ be the operator defined by $\text{ParallelRSolve}\left(\bm{M}_0, \bm{M}_1,\ldots, \bm{M}_d, \bm{b}_0\right)$, namely, $\bm{x}_0 = \bm{Z}_0\bm{b}_0$. Then
\begin{equation}\label{appr_inv_express}
\bm{Z}_0\approx_{\sum_{i=0}^{d}\epsilon_i}\bm{M}^{-1}_0
\end{equation}
\end{lemma}


Algorithm~\ref{Algo:Inv} returns a ``crude'' solution to $\bm{M}_{0}\bm{x}=\bm{b}$. To obtain arbitrary close solutions, Spielman {\it et.al}~\cite{Spielman} introduced the \emph{preconditioned Richardson iterative scheme}, summarized in Algorithm~\ref{Algo:Inv2}. Following their analysis, Lemma~\ref{Exact_Alg_guarantee_lemma} provides the iteration count needed by Algorithm~\ref{Algo:Inv2} to arrive at $\tilde{\bm{x}}$. 
\begin{algorithm}[h!]
  \caption{$\text{ParallelESolve}\left(\bm{M}_0, \bm{M}_1,\ldots, \bm{M}_d,  \bm{b}_0, \epsilon\right)$}
  \label{Algo:Inv2}
  \begin{algorithmic}[1]
	\State \textbf{Input}: Inverse approximated chain $\{\bm{M}_0,\bm{M}_1,\ldots, \bm{M}_d\}$, $\bm{b}_0$, and $\epsilon$. 
	\State \textbf{Output}: $\epsilon$ close approximation, $\tilde{\bm{x}}$, of $\bm{x}^*$
	\State \textbf{Initialize}: $\bm{y}_0 = 0$; \\
	$\chi = \text{ParallelRSolve}\left(\bm{M}_0,\bm{M}_1,\ldots, \bm{M}_d, \bm{b}_0\right)$ (i.e., Algorithm~\ref{Algo:Inv})
    \For  {$k=1$ to $q$}
    	\State $\bm{u}_{k}^{(1)} = \bm{M}_0\bm{y}_{k-1}$
    	\State $\bm{u}_{k}^{(2)} = \text{ParallelRSolve}\left(\bm{M}_0, \bm{M}_1,\ldots, \bm{M}_d, \bm{u}_{k}^{(1)}\right)$
    	\State $\bm{y}_{k} = \bm{y}_{k-1} - \bm{u}_{k}^{(2)} + \chi$ 
    \EndFor \textbf{end for} 
    \State $\tilde{\bm{x}} = \bm{y}_q$
    \State \textbf{return} $\tilde{\bm{x}}$ 
  \end{algorithmic}
\end{algorithm}

\begin{lemma}~\cite{Spielman}\label{Exact_Alg_guarantee_lemma}
Let $\{\bm{M}_0, \bm{M}_1\ldots \bm{M}_d\}$ be an inverse approximated chain such that $\sum_{i=1}^{d}\epsilon_i < \frac{1}{3}\ln2$. Then $\text{ParallelESolve}\left(\bm{M}_0, \bm{M}_1,\ldots, \bm{M}_d, \bm{b}_0, \epsilon\right)$ arrives at an $\epsilon$ close solution of $\bm{x}^{\star}$ in $q = \mathcal{O}\left(\log\frac{1}{\epsilon}\right)$ iterations. 
\end{lemma}


\section{Distributed SDDM Solvers}\label{Sec:FullVersion}
Next, we distribute the parallel solver of Section~\ref{Sec:ParrallelSolver}. Similar to~\cite{Spielman}, we first introduce an approximate inverse chain which can be computed in a distributed fashion. This leads us to distributed version of the ``crude'' solver (see Section~\ref{Sec:RudeSolver}). Contrary to~\cite{Spielman}, however, we then generalize the ``crude'' distributed solver to allow for \emph{exact} solutions (see Section~\ref{Sec:Exact}) of Equation~\ref{lin_sys}. We summarize our results in the following theorem: 
\begin{theorem}\label{Main_Theorem}
There exists a distributed algorithm,
$\mathcal{A}\left(\{[\bm{M}_0]_{k1},\ldots [\bm{M}_0]_{kn}\}, [\bm{b}_0]_k, \epsilon\right)$, that computes $\epsilon$-close approximations to the solution of $\bm{M}_{0}\bm{x}=\bm{b}_{0}$ in $\mathcal{O}\left(n^2\log\kappa \log\left(\frac{1}{\epsilon}\right)\right)$ time steps, with $n$ the number of nodes in $\mathbb{G}$, $\kappa$ the condition number of  $\bm{M}_0$, and $[\bm{M}_{0}]_{k\cdot}$ the $k^{th}$ row of $\bm{M}_{0}$, as well as $\epsilon\in \left(0, \frac{1}{2}\right]$ representing the precision parameter. 
\end{theorem}

Note that for each node $\bm{v}_k\in \mathbb{V}$, the input information for algorithm $\mathcal{A}$ is the $k^{th}$ row of $\bm{M}_0$ (i.e., the weights of the edges adjacent to $\bm{v}_i$), the precision parameter $\epsilon$, and the $k^{th}$ component of $[\bm{b}_0]$ (i.e., $[\bm{b}_{0}]_{k}$), easily rendering a distributed solver. 

\subsection{``Crude'' Distributed SDDM Solver}\label{Sec:RudeSolver}
Starting from $\bm{M}_0 = \bm{D}_0 - \bm{A}_0$, consider the collection $\mathcal{C} = \{\bm{A}_0,\bm{D}_0,\bm{A}_1, \bm{D}_1,\ldots, \bm{A}_d, \bm{D}_d\}$, where $\bm{D}_{k}=\bm{D}_{0}$, and $\bm{A}_{k}=\bm{D}_{0}\left(\bm{D}_{0}^{-1}\bm{A}_{0}\right)^{2^{k}}$, for $k=\{1,\dots,d\}$ with $\bm{D}_{0}=\bm{D}_{0}$, and $\bm{A}_{0}=\bm{A}_{0}$. Since the magnitude of the eigenvalues of $\bm{D}^{-1}_0\bm{A}_0$ is strictly less than 1, $\left(\bm{D}^{-1}_0\bm{A}_0\right)^{2^{k}}$ tends to zero as $k$ increases which reduces the length of the chain needed for the distributed solver.  It is easy to verify that $\mathcal{C}$ is an inverse approximated chain, since: (1) $\bm{D}_i - \bm{A}_i \approx_{\epsilon_{i-1}} \bm{D}_{i-1} - \bm{A}_{i-1}\bm{D}^{-1}_{i-1}\bm{A}_{i-1}$ with $\epsilon_i = 0$ for $i = 1,\ldots, d$, (2) $\bm{D}_{i}\approx_{\epsilon_{i-1}}\bm{D}_{i-1}$ with $\epsilon_i = 0$ for $i = 1,\ldots, d$, and (3) $\bm{D}_d\approx_{\epsilon_d}\bm{D}_d - \bm{A}_d$. Using the above, Algorithm~\ref{Algo:DisRudeApprox} (our first contribution) describes the distributed version of the ``crude'' parallel solver, which returns the $k^{th}$ component of the approximate solution vector, $[\bm{x}_{0}]_{k}$. Each node, $\bm{v}_k\in \mathbb{V}$, receives the $k^{th}$ row of $\bm{M}_0$, the $k^{th}$ value of $\bm{b}_{0}$ (i.e., $[\bm{b}_0]_k$), and the length of the inverse approximated chain $d$ as inputs. It operates in two parts, \textbf{Part One} and \textbf{Part Two}. The first, computes the $k^{th}$ component of $\bm{b}_{i}$, $[\bm{b}_{i}]_{k}$ for $i\in\{1,\dots,d\}$, which is then used in \textbf{Part Two} to return $[\bm{x}_{0}]_{k}$.

\begin{algorithm}[t!]
  \caption{
 		 \hspace{0em} $\text{DistrRSolve}\Big(\left\{[\bm{M}_0]_{k1},\ldots, [\bm{M}_0]_{kn}\right\}, 	 [\bm{b}_0]_k, d\Big)$
	   }
	   \label{Algo:DisRudeApprox}
  \begin{algorithmic}
	\State \textbf{Part One: Computing $[\bm{b}_{i}]_{k}$}
	\State $[\bm{b}_1]_k = [\bm{b}_0]_k + \sum_{j: \bm{v}_j\in \mathbb{N}_{1}\left(\bm{v}_k\right)}[\bm{A}_0\bm{D}^{-1}_0]_{kj}[\bm{b}_{0}]_j$	
	\For {$i = 2$ to $d$}
		\For {$j: \bm{v}_j \in \mathbb{N}_{2^{i-1}}\left(\bm{v}_k\right)$}
			\State$
			\left[(\bm{A}_0\bm{D}^{-1}_0)^{2^{i-1}}\right]_{kj} =  \sum_{r=1}^{n}\frac{[\bm{D}_0]_{rr}}{[\bm{D}_0]_{jj}}\left[(\bm{A}_0\bm{D}^{-1}_0)^{2^{i-2}}\right]_{kr} \left[(\bm{A}_0\bm{D}^{-1}_0)^{2^{i-2}}\right]_{jr}$
		\EndFor \textbf{end for}		
		\State $[\bm{b}_i]_k = [\bm{b}_{i-1}]_k + \sum_{j: \bm{v}_j \in \mathbb{N}_{2^{i-1}}\left(\bm{v}_k\right)}\left[(\bm{A}_0\bm{D}^{-1}_0)^{2^{i-1}}\right]_{kj}[\bm{b}_{i-1}]_{j}$
	\EndFor \textbf{end for}
	\\\hrulefill
	\State \textbf{Part Two: Computing $[\bm{x}_{0}]_{k}$}
	\State $[\bm{x}_d]_k = \sfrac{[\bm{b}_d]_k}{[\bm{D}_0]_{kk}}$
	\For {$i = d - 1$ to $1$}
		\For {$j: \bm{v}_j \in \mathbb{N}_{2^i}(\bm{v}_k)$}
			\State$
			\left[(\bm{D}^{-1}_0\bm{A}_0)^{2^i}\right]_{kj} = 
			\sum_{r=1}^{n}\frac{[\bm{D}_0]_{jj}}{[\bm{D}_0]_{rr}}\left[(\bm{D}^{-1}_0\bm{A}_0)^{2^{i-1}}\right]_{kr} \left[(\bm{D}^{-1}_0\bm{A}_0)^{2^{i-1}}\right]_{jr}$
		\EndFor \textbf{end for}		
		\State$
		[\bm{x}_i]_k =  \frac{[\bm{b}_i]_k}{2[\bm{D}_0]_{kk}} + \frac{[\bm{x}_{i+1}]_{k+1}}{2} +\frac{1}{2}\sum_{j: \bm{v}_j \in \mathbb{N}_{2^{i}}(\bm{v}_k)}\left[(\bm{D}^{-1}_0\bm{A}_0)^{2^i}\right]_{kj}[\bm{x}_{i+1}]_j$
	\EndFor \textbf{end for}
	\State $[\bm{x}_0]_k = \frac{[\bm{b}_0]_k}{2[\bm{D}_{0}]_{kk}} + \frac{[\bm{x}_{1}]_k}{2} + \frac{1}{2}\sum_{j:\bm{v}_j\in \mathbb{N}_{1}(\bm{v}_k)}[\bm{D}^{-1}_0\bm{A}_0]_{kj}[\bm{x}_1]_j$
    \State \textbf{return:} $[\bm{x}_0]_k$ 
  \end{algorithmic}
\end{algorithm}
\begin{algorithm}[h!]
  \caption{$\text{DistrESolve}\left(\{[\bm{M}_0]_{k1},\ldots, [\bm{M}_0]_{kn}\}, [\bm{b}_0]_k, d, \epsilon\right)$}

  \begin{algorithmic}
\State \textbf{Initialize}: $[\bm{y}_0]_k = 0$; 
$[\chi]_k = \text{DistrRSolve}\left(\{[\bm{M}_0]_{k1},\ldots, [\bm{M}_0]_{kn}\}, [\bm{b}_0]_k, d\right)$ (i.e., Algorithm~\ref{Algo:DisRudeApprox})	
	\For {$t=1$ to $q$}
		\State $\left[\bm{u}^{(1)}_{t}\right]_k = [\bm{D}_0]_{kk}[\bm{y}_{t-1}]_k - \sum_{j: \bm{v}_j\in \mathbb{N}_{1}(\bm{v}_k)}[\bm{A}_{0}]_{kj}[\bm{y}_{t-1}]_j$
		\State $\left[\bm{u}^{(2)}_{t}\right]_k = \text{DistrRSolve}(\{[\bm{M}_0]_{k1},\ldots, [\bm{M}_0]_{kn}\}, \left[\bm{u}^{(1)}_{t}\right]_k, d,)$
		\State $[\bm{y}_t]_k = [\bm{y}_{t-1}]_k - \left[\bm{u}^{(2)}_{t}\right]_k + [\chi]_k$ 
	\EndFor \textbf{end for}	
    \State $[\tilde{\bm{x}}]_k = [\bm{y}_q]_k$
    \State \textbf{return} $[\tilde{\bm{x}}]_k$ 
  \end{algorithmic}
  \label{Alg_ExactBla}  
\end{algorithm}

\textbf{Analysis of Algorithm~\ref{Algo:DisRudeApprox}:} Next, we present the theoretical analysis, showing that $\text{DistrRSolve}$ computes the $k^{th}$ component of the ``crude'' approximation of $\bm{x}^{\star}$. Further, we provide the time complexity analysis. 

\begin{lemma}\label{Rude_Dec_Alg_guarantee_Lemma}
Let $\bm{M}_0 = \bm{D}_0 - \bm{A}_0$ be the standard splitting of $\bm{M}_{0}$. Let $\bm{Z}^{\prime}_0$ be the operator defined by $\text{DistrRSolve}([\{[\bm{M}_0]_{k1},\ldots, [\bm{M}_0]_{kn}\}, [\bm{b}_0]_k, d)$ (i.e., $\bm{x}_0 = \bm{Z}^{\prime}_0\bm{b}_0$). Then
\begin{equation*}
\bm{Z}^{\prime}_0\approx_{\epsilon_d} \bm{M}^{-1}_0
\end{equation*}
Moreover, Algorithm~\ref{Algo:DisRudeApprox} requires  $\mathcal{O}\left(dn^2\right)$ time steps. 
\end{lemma}
\begin{proof}
See Appendix. 
\end{proof}
\subsection{``Exact'' Distributed SDDM Solver}\label{Sec:Exact}
Having introduced $\text{DistrRSolve}$, we are now ready to present a distributed version of Algorithm~\ref{Algo:Inv2} which enables the computation of $\epsilon$ close solutions for $\bm{M}_{0}\bm{x}=\bm{b}_{0}$.  Similar to $\text{DistrRSolve}$, each node $\bm{v}_k\in \mathbb{V}$ receives the $k^{th}$ row of $\bm{M}_0$, $[\bm{b}_0]_k$, $d$ and a precision parameter $\epsilon$ as inputs. Node $\bm{v}_k$ then computes the $k^{th}$ component of the $\epsilon$  close approximation of $\bm{x}^{\star}$.  

\textbf{Analysis of Algorithm~\ref{Alg_ExactBla}:}
The following lemma shows that $\text{DistrESolve}$ computes the $k^{th}$ component of the $\epsilon$ close approximation of $x^{\star}$ and provides the time complexity analysis

\begin{lemma}\label{Dist_Exact_algorithm_guarantee_lemma}
Let $\bm{M}_0 = \bm{D}_0 - \bm{A}_0$ be the standard splitting. Further, let $\epsilon_d < \frac{1}{3}\ln2$ in the nverse approximated chain $\mathcal{C} = \{\bm{A}_0,\bm{D}_0,\bm{A}_1, \bm{D}_1,\ldots, \bm{A}_d, \bm{D}_d\}$. Then $\text{DistrESolve}(\{[\bm{M}_0]_{k1},\ldots, [\bm{M}_0]_{kn}\}, [b_0]_k, d, \epsilon)$ requires $\mathcal{O}\left(\log\frac{1}{\epsilon}\right)$ iterations to return the $k^{th}$ component of the $\epsilon$ close approximation for $x^{\star}$.
\end{lemma}
\begin{proof}
See Appendix. 
\end{proof}

The following lemma provides the time complexity analysis of $\text{DistrESolve}$:

\begin{lemma}\label{time_complexity_of_distresolve}
Let $\bm{M}_0 =\bm{D}_0 - \bm{A}_0$ be the standard splitting. Further, let $\epsilon_d < \frac{1}{3}\ln2$ in the inverse approximated chain $\mathcal{C} = \{\bm{A}_0,\bm{D}_0,\bm{A}_1, \bm{D}_1,\ldots, \bm{A}_d, \bm{D}_d\}$. Then, $\text{DistrESolve}(\{[\bm{M}_0]_{k1},\ldots, [\bm{M}_0]_{kn}\}, [\bm{b}_0]_k, d, \epsilon)$ requires $\mathcal{O}\left(dn^2\log(\frac{1}{\epsilon})\right)$ time steps.
\end{lemma}

\begin{proof}
See Appendix. 
\end{proof}
\subsection{Length of the Inverse Chain}\label{Sec:Length}
Both introduced algorithms depend on the length of the inverse approximated chain, $d$. Here, we provide an analysis to determine the value of $d$ which guarantees $\epsilon_d < \frac{1}{3}\ln2$ in $\mathcal{C} = \{\bm{A}_0,\bm{D}_0,\bm{A}_1, \bm{D}_1,\ldots, \bm{A}_d, \bm{D}_d\}$:

\begin{lemma}\label{eps_d_lemma}
Let $\bm{M}_0 = \bm{D}_0 - \bm{A}_0$ be the standard splitting and $\kappa$ denote the condition number of $\bm{M}_0$. Consider the inverse approximated chain $\mathcal{C} = \{\bm{A}_0,\bm{D}_0,\bm{A}_1, \bm{D}_1,\ldots, \bm{A}_d, \bm{D}_d\}$ with a length $ d =  \lceil \log \left(2\ln\left(\frac{\sqrt[3]{2}}{\sqrt[3]{2} - 1}\right)\kappa\right)\rceil$, then $\bm{D}_0\approx_{\epsilon_d} \bm{D}_0 - \bm{D}_0\left(\bm{D}^{-1}_0\bm{A}_0\right)^{2^d}$, with $\epsilon_d < \frac{1}{3}\ln2$.
\end{lemma}

\begin{proof}
The proof will be given as a collection of claims:
\begin{claim}
Let $\kappa$ be the condition number of $\bm{M}_0 = \bm{D}_0 - \bm{A}_0$, and $\{\lambda_i\}^{n}_{i=1}$ denote the eigenvalues of $\bm{D}^{-1}_0\bm{A}_0$. Then,
$|\lambda_i| \le 1 - \frac{1}{\kappa}$,  
for all $i=1,\ldots, n$
\end{claim}
\begin{proof}
See Appendix. 
\end{proof}

Notice that if $\lambda_i$ represented an eigenvalue of $\bm{D}^{-1}_0\bm{A}_0$, then $\lambda^{r}_i$ is an eigenvalue of $\left(\bm{D}^{-1}_0\bm{A}_0\right)^r$ for all $r\in \mathbb{N}$. Therefore, we have 
\begin{equation}\label{spect_radius}
\rho\left(\left(\bm{D}^{-1}_0\bm{A}_0\right)^{2^d}\right)\le \left(1 - \frac{1}{\kappa}\right)^{2^d}
\end{equation}

\begin{claim}
Let $\bm{M}$ be an SDDM matrix and consider the splitting $\bm{M} =\bm{D} -\bm{A}$, with $\bm{D}$ being non negative diagonal and $\bm{A}$ being symmetric non negative. Further, assume that the eigenvalues of $\bm{D}^{-1}\bm{A}$ lie between $-\alpha$ and $\beta$. Then, $
(1 - \beta)\bm{D} \preceq \bm{D} -\bm{A} \preceq (1 + \alpha)\bm{D}$. 
\end{claim}
\begin{proof}
See Appendix. 
\end{proof}

Combining the above results, give
$\left[1 - \left(1 - \frac{1}{\kappa}\right)^{2^d}\right]\bm{D}_d\preceq \bm{D}_d - \bm{A}_d\preceq\left[1 + \left(1 - \frac{1}{\kappa}\right)^{2^d}\right]\bm{D}_d $
Hence, to guarantee that $\bm{D}_d\approx_{\epsilon_d}\bm{D}_d - \bm{A}_d$, the following system must be satisfied: (1) $e^{-\epsilon_d} \le 1 - \left(1 - \frac{1}{\kappa}\right)^{2^d}$, and (2)  
$e^{\epsilon_d} \ge 1 + \left(1 - \frac{1}{\kappa}\right)^{2^d}$. Introducing $\gamma$ for $\left(1 - \frac{1}{\kappa}\right)^{2^d}$, we arrive at: (1) 
$\epsilon_d \ge \ln\left(\frac{1}{1 - \gamma}\right)$, and (2) 
$\epsilon_d \ge \ln(1 + \gamma)$. Hence, $\epsilon_d \ge \max\left\lbrace \ln\left(\frac{1}{1 - \gamma}\right), \ln(1 + \gamma)\right\rbrace = \ln\left(\frac{1}{1 - \gamma}\right)$. 
Now, notice that if $d = \lceil \log c \kappa\rceil$ then,
$\gamma = \left(1 - \frac{1}{\kappa}\right)^{2^d} =  \left(1 - \frac{1}{\kappa}\right)^{{c}\kappa} \le \frac{1}{e^c}$. 
Hence, $\ln\left(\frac{1}{1 - \gamma}\right) \le \ln\left(\frac{e^c}{e^c - 1}\right)$. This gives $c = \lceil 2\ln\left(\frac{\sqrt[3]{2}}{\sqrt[3]{2} - 1}\right)\rceil$, implying $\epsilon_d = \ln\left(\frac{e^c}{e^c - 1}\right) < \frac{1}{3}\ln2$.
\end{proof}
Using the above results the time complexity of $\text{DistrESolve}$ with $d =  \lceil \log \left(2\ln\left(\frac{\sqrt[3]{2}}{\sqrt[3]{2} - 1}\right)\kappa\right)\rceil$ is $\mathcal{O}\left(n^2\log\kappa \log(\frac{1}{\epsilon})\right)$ times steps, which concludes the proof of Theorem \ref{Main_Theorem}.

\section{Distributed R-Hop SDDM Solver}
Though the previous algorithm requires no knowledge of the graph's topology, but it requires the information of all other nodes (i.e., full communication). We will outline an R-Hop version of the algorithm in which communication is restricted to the R-Hop neighborhood between nodes. The following theorem summarizes these main results:
\begin{theorem}\label{Main_TheoremTwo}
There is a decentralized algorithm 
$\mathcal{A}(\{[\bm{M}_0]_{k1},\ldots [\bm{M}_0]_{kn}\}, [\bm{b}_0]_k, R, \epsilon)$, that uses only $R$-Hop communication between the nodes and computes $\epsilon$-close solutions to $\bm{M}_{0}\bm{x}=\bm{b}_{0}$ in $\mathcal{O}\left(\left(\frac{\alpha\kappa}{R} + \alpha Rd_{max}\right) \log(\frac{1}{\epsilon})\right)$ time steps, with $n$ being the number of nodes in $\mathbb{G}$, $d_{max}$ denoting the maximal degree, $\kappa$ the condition number of $\bm{M}_0$, and $\alpha = \min\left\{n, \frac{\left(d^{R+1}_{\text{max}} - 1\right)}{\left(d_{\text{max}} - 1\right)}\right\}$ representing the upper bound on the size of the R-hop neighborhood $\forall \bm{v}\in \mathbb{V}$, and $\epsilon\in (0, \frac{1}{2}]$ being the precision parameter. 
\end{theorem}

Given a graph $\mathbb{G}$ formed from the weighted Laplacian $\bm{M}_{0}$, the following corollary easily follows: 

\begin{corollary}\label{CorollaryTwo}
Let $\bm{M}_0$ be the weighted Laplacian of $\mathbb{G} = \left(\mathbb{V},\mathbb{E},\bm{W}\right)$. There exists a decentralized algorithm that uses only $R$-hop communication between nodes and computes $\epsilon$ close solutions of $\bm{M}_{0}\bm{x}=\bm{b}_{0}$ in
 $\mathcal{O}\left(\frac{n^3\alpha}{R}\frac{\bm{W}_{\text{max}}}{\bm{W}_{\text{min}}}\log(\frac{1}{\epsilon})\right)$ time steps, with $n$ being the number of nodes in $\mathbb{G}$, $\bm{W}_{\text{max}}, \bm{W}_{\text{min}}$ denoting the largest and the smallest weights of edges in $\mathbb{G}$, respectively, $\alpha = \min\left\{n, \frac{\left(d^{R+1}_{\text{max}} - 1\right)}{\left(d_{\text{max}} - 1\right)}\right\}$ representing the upper bound on the size of the R-hop neighborhood $\forall \bm{v}\in \mathbb{V}$, and $\epsilon\in (0, \frac{1}{2}]$ being the precision parameter. 
\end{corollary}


\subsection{``Crude'' R-Hop SDMM Solver}
Algorithm~\ref{Algo:DistRHop} presents the ``crude'' R-Hop solver for SDDM systems using the same inverse chain of Section~\ref{Sec:RudeSolver}. Each node $\bm{v}_{k} \in \mathbb{V}$ receives the $k^{th}$ row of $\bm{M}_{0}$ , $k^{th}$ component, $[\bm{b}_{0}]_{k}$ of $\bm{b}_{0}$, the length of the inverse chain, $d$, and the local communication bound\footnote{For simplicity, $R$ is assumed to be in the order of powers of 2, i.e., $R=2^{\rho}$.}  $R$ as inputs, and outputs the $k^{th}$ component of the ``rude'' approximation of $\bm{x}^{\star}$.

\begin{algorithm}[h!]
  \caption{$\text{RDistRSolve}\left(\{[\bm{M}_0]_{k1},\ldots, [\bm{M}_0]_{kn}\}, [\bm{b}_0]_k, d, R\right)$}
  \label{Algo:DistRHop}
  \begin{algorithmic}
	\State \textbf{Part One:}	
	\State $\{[\bm{A}_0\bm{D}^{-1}_0]_{k1},\ldots,[\bm{A}_0\bm{D}^{-1}_0]_{kn}\} = \left\{\frac{[\bm{A}_0]_{k1}}{[\bm{D}_0]_{11}},\ldots,\frac{[\bm{A}_0]_{kn}}{[\bm{D}_0]_{nn}}\right\}$, $\{[\bm{D}^{-1}_0\bm{A}_0]_{k1},\ldots,[\bm{D}^{-1}_0\bm{A}_0]_{kn}\} = \{\frac{[\bm{A}_0]_{k1}}{[\bm{D}_0]_{kk}},\ldots,\frac{[\bm{A}_0]_{kn}}{[\bm{D}_0]_{kk}}\}$
	\State $[\bm{C}_0]_{k1},\ldots,[\bm{C}_0]_{kn} = \text{Comp}_0\left([\bm{M}_0]_{k1},\ldots, [\bm{M}_0]_{kn}, R\right)$, $[\bm{C}_1]_{k1},\ldots,[\bm{C}_1]_{kn}= \text{Comp}_1\left([\bm{M}_0]_{k1},\ldots, [\bm{M}_0]_{kn}, R\right)$
\\\hrulefill	
	\State \textbf{Part Two:}
	\For {$i=1$ to $d$}
		\State \textbf{if} \ {$i-1 < \rho$}
			\State $[\bm{u}^{(i-1)}_{1}]_k = [\bm{A}_0\bm{D}^{-1}_0\bm{b}_{i-1}]_k$
			\For {$j=2$ to $2^{i-1}$}
				\State $[\bm{u}^{(i-1)}_{j}]_k = [\bm{A}_0\bm{D}^{-1}_0\bm{u}^{(i-1)}_{j-1}]_k$
			\EndFor \textbf{end for}
			\State $[\bm{b}_i]_k = [\bm{b}_{i-1}]_k + [\bm{u}^{(i-1)}_{2^{i-1}}]_k$
		\State \textbf{if} {$i-1 \ge \rho$}
			\State $l_{i-1} = \sfrac{2^{i-1}}{R}$
			\State $[\bm{u}^{(i-1)}_{1}]_k = [\bm{C}_0\bm{b}_{i-1}]_k$
			\For {$j=2$ to $l_{i-1}$}
				\State $[\bm{u}^{(i-1)}_{j}]_k = [\bm{C}_0\bm{u}^{(i-1)}_{j-1}]_k$
			\EndFor \textbf{end for}
			\State $[\bm{b}_i]_k = [\bm{b}_{i-1}]_k + [\bm{u}^{(i-1)}_{l_{i-1}}]_k$
	\EndFor \textbf{end for}
		\\\hrulefill
	\State \textbf{Part Three:}
	\State $[\bm{x}_d]_k = \sfrac{[\bm{b}_d]_k}{[\bm{D}_0]_{kk}}$
	\For {$i=d-1$ to $1$}
		\State \textbf{if} {$i < \rho$}
			\State $[\bm{\eta}^{(i+1)}_{1}]_k = [\bm{D}^{-1}_0\bm{A}_0\bm{x}_{i+1}]_k$
			\For {$j=2$ to $2^{i}$}
				\State $[\bm{\eta}^{(i+1)}_{j}]_k = [\bm{D}^{-1}_0\bm{A}_0\bm{\eta}^{(i+1)}_{j-1}]_k$
			\EndFor \textbf{end for}
			\State $[\bm{x}_i]_k = \frac{1}{2}\left[\frac{[\bm{b}_i]_k}{[\bm{D}_{0}]_{kk}} + [\bm{x}_{i+1}]_k + [\bm{\eta}^{i + 1}_{2^i}]_{k}\right]$
		\State \textbf{if} {$i \ge \rho$}
			\State $l_i = \sfrac{2^i}{R}$
			\State $[\bm{\eta}^{(i+1)}_{1}]_k = [\bm{C}_1\bm{x}_{i+1}]_k$
			\For {$j=2$ to $l_i$}
				\State $[\bm{\eta}^{(i+1)}_{j}]_k = [\bm{C}_1\bm{\eta}^{(i+1)}_{j-1}]_k$
			\EndFor \textbf{end for}
			\State $[\bm{x}_i]_k = \frac{1}{2}\left[\frac{[\bm{b}_i]_k}{[\bm{D}_0]_{kk}}  + [\bm{x}_{i+1}]_k + [\bm{\eta}^{i+1}_{l_i}]_k \right]$
	\EndFor \textbf{end for}
	\State $[\bm{x}_0]_k = \frac{1}{2}\left[\frac{[\bm{b}_0]_k}{[\bm{D}_0]_{kk}} + [\bm{x}_1]_k + [\bm{D}^{-1}_0\bm{A}_0\bm{x}_1]_k \right]$
    \State \textbf{return} $[\bm{x}_0]_k$ 
  \end{algorithmic}
\end{algorithm}
\begin{algorithm}[h!]
  \caption{$\text{Comp}_0\left([\bm{M}_0]_{k1},\ldots, [\bm{M}_0]_{kn}, R\right)$}
  \label{Alg_4}
  \begin{algorithmic}
	\For {$l = 1$ to $R - 1$}
		\For {$j$ s.t.$\bm{v}_j\in \mathbb{N}_{l+1}(\bm{v}_k)$}
			\State
			$\left[(\bm{A}_0\bm{D}^{-1}_0)^{l+1}\right]_{kj} = \sum\limits_{r:\bm{v}_r\in \mathbb{N}_1(v_j)}\frac{[\bm{D}_0]_{rr}}{[\bm{D}_0]_{jj}}[(\bm{A}_0\bm{D}^{-1}_0)^l]_{kr}[\bm{A}_0\bm{D}^{-1}_0]_{jr}$
		\EndFor\textbf{end for}
	\EndFor\textbf{end for}
	\State \textbf{return} $\bm{c}_0 = \{[(\bm{A}_0\bm{D}^{-1}_0)^{R}]_{k1},\ldots,[(\bm{A}_0\bm{D}^{-1}_0)^{R}]_{kn} \}$
  \end{algorithmic}
\end{algorithm}
\begin{algorithm}
  \caption{$\text{Comp}_1([\bm{M}_0]_{k1},\ldots, [\bm{M}_0]_{kn}, R)$}
  \label{Alg_5}
  \begin{algorithmic}
	\For {$l = 1$ to $R - 1$}
		\For {$j$ s.t.$\bm{v}_j\in \mathbb{N}_{l+1}(\bm{v}_k)$}
			\State$
			\left[(\bm{D}^{-1}_0\bm{A}_0)^{l+1}\right]_{kj} = \sum\limits_{r:\bm{v}_r\in \mathbb{N}_1(\bm{v}_j)}\frac{[\bm{D}_0]_{jj}}{[\bm{D}_0]_{rr}}[(\bm{D}^{-1}_0\bm{A}_0)^l]_{kr}[\bm{D}^{-1}_0\bm{A}_0]_{jr}$
		\EndFor\textbf{end for}
	\EndFor\textbf{end for}
	\State \textbf{return} $\bm{c}_1 = \{[(\bm{D}^{-1}_0\bm{A}_0)^{R}]_{k1},\ldots,[(\bm{D}^{-1}_0\bm{A}_0)^{R}]_{kn} \}$
  \end{algorithmic}
\end{algorithm}

\textbf{Analysis of Algorithm~\ref{Algo:DistRHop}} The following Lemma shows that $\text{RDistRSolve}$ computes the $k^{th}$ component of the ``crude'' approximation of $\bm{x}^{\star}$ and provides the algorithm's time complexity 
 
\begin{lemma}\label{r_hop_rude_lemma}
Let $\bm{M}_0 = \bm{D}_0 - \bm{A}_0$ be the standard splitting and let $\bm{Z}^{\prime}_0$ be the operator defined by $\text{RDistRSolve}$, namely, $\bm{x}_0 = \bm{Z}^{\prime}_0\bm{b}_0$. Then,
$\bm{Z}^{\prime}_0\approx_{\epsilon_d} \bm{M}^{-1}_0$.  
$\text{RDistRSolve}$ requires $\mathcal{O}\left(\frac{2^d}{R}\alpha + \alpha Rd_{max}\right)$, where $\alpha = \min\left\{n, \frac{\left(d^{R+1}_{\text{max}} - 1\right)}{\left(d_{\text{max}} - 1\right)}\right\}$, to arrive at $\bm{x}_{0}$. 
\end{lemma}
The proof of the above Lemma can be arrived at by proving a collection of claims:
\begin{claim}\label{claim_1}
Matrices $\left(\bm{D}^{-1}_0\bm{A}_0\right)^r$ and $\left(\bm{A}_0\bm{D}^{-1}_0\right)^{-r}$ have sparsity patterns corresponding to the $R$-Hop neighborhood for any $R\in \mathbb{N}$.
\end{claim}
\begin{proof}
The above claim is proved by induction on $R$. We start with the base case: for $R = 1$
\begin{equation*}
[\bm{A}_0\bm{D}^{-1}_0]_{ij} =\frac{[\bm{A}_0]_{ij}}{[\bm{D}_0]_{ii}} \ \ \  (\text{if } j: \bm{v}_j\in \mathbb{N}_{1}(\bm{v}_i))  \ \ \ \ \ \ \text{or} \ \ \ \ [\bm{A}_0\bm{D}^{-1}_0]_{ij}=0 \ \ \ (\text{otherwise})
\end{equation*}
Therefore, $\bm{A}_0\bm{D}^{-1}_0$ has a sparsity pattern corresponding to the $1$-Hop neighborhood. 
Assume that for all $1\le p\le R-1$, $\left(\bm{A}_0\bm{D}^{-1}_0\right)^p$ has a sparsity pattern corresponding to the $p-hop$ neighborhood. Consider, $\left(\bm{A}_0\bm{D}^{-1}_0\right)^{r}$
\begin{equation}\label{sparisty_lemma_2}
[(\bm{A}_0\bm{D}^{-1}_0)^{R}]_{ij} = \sum_{k=1}^{n}[(\bm{A}_0\bm{D}^{-1}_0)^{R-1}]_{ik}[\bm{A}_0\bm{D}^{-1}_0]_{kj}
\end{equation}
Since $\bm{A}_0\bm{D}^{-1}_0$ is non negative, then $[(\bm{A}_0\bm{D}^{-1}_0)^{R}]_{ij} \ne 0$ iff there exists $k$ such that $\bm{v}_k\in \mathbb{N}_{R-1}(\bm{v}_i)$ and $\bm{v}_k\in \mathbb{N}_1(\bm{v}_j)$, namely, $\bm{v}_j\in \mathbb{N}_{R}(\bm{v}_i)$. The proof can be done in a similar fashion for $\bm{D}^{-1}_0\bm{A}_0$.
\end{proof}
The next claim provides complexity guarantees for $\text{Comp}_{0}$ and $\text{Comp}_{1}$ described in Algorithms~\ref{Alg_4} and~\ref{Alg_5}, respectively. 
\begin{claim}\label{claim_2}
Algorithms~\ref{Alg_4} and~\ref{Alg_5} use only the R-hop information to compute the $k^{th}$ row of $\left(\bm{D}^{-1}_0\bm{A}_0\right)^{R}$ and $\left(\bm{A}_0\bm{D}^{-1}_0\right)^{R}$, respectively, in $\mathcal{O}\left(\alpha Rd_{max}\right)$ time steps, where $\alpha = \min\left\{n, \frac{\left(d^{R+1}_{\text{max}} - 1\right)}{\left(d_{\text{max}} - 1\right)}\right\}$.  
\end{claim}
\begin{proof}
The proof will be given for $\text{Comp}_{0}$ described in Algorithm~\ref{Alg_4} as that for $\text{Comp}_{1}$ can be performed similarly. Due to Claim~\ref{claim_1}, we have
\begin{align}\label{eq_1}
\left[\left(\bm{A}_0\bm{D}^{-1}_0\right)^{l+1}\right]_{kj} = \sum\limits_{r=1}^{n}\left[\left(\bm{A}_0\bm{D}^{-1}_0\right)^{l}\right]_{kr}\left[\bm{A}_0\bm{D}^{-1}_0\right]_{rj} = \sum\limits_{r: \bm{v}_r\in \mathbb{N}_1(\bm{v}_j)}\left[\left(\bm{A}_0\bm{D}^{-1}_0\right)^{l}\right]_{kr}\left[\bm{A}_0\bm{D}^{-1}_0\right]_{rj}
\end{align} 
Therefore at iteration $l+1$, $\bm{v}_k$ computes the $k^{th}$ row of $\left(\bm{A}_0\bm{D}^{-1}_0\right)^{l+1}$ using: (1) the $k^{th}$ row of $(\bm{A}_0\bm{D}^{-1}_0)^{l}$, and (2) the $r^{th}$ column of $\bm{A}_0\bm{D}^{-1}_0$. Node $\bm{v}_{r}$, however, can only send the $r^{th}$ row of $\bm{A}_{0}\bm{D}^{-1}_{0}$ making $\bm{A}_{0}\bm{D}^{-1}_{0}$ non-symmetric. Noting that
$\sfrac{[\bm{A}_0\bm{D}^{-1}_0]_{rj}}{[\bm{D}_0]_{rr}} = \sfrac{[\bm{A}_0\bm{D}^{-1}_0]_{jr}}{[\bm{D}_0]_{jj}}$, since $\bm{D}^{-1}_0\bm{A}_0\bm{D}^{-1}$ is symmetric, leads to 
$[(\bm{A}_0\bm{D}^{-1}_0)^{l+1}]_{kj} = 
\sum\limits_{r: \bm{v}_r\in \mathbb{N}_1(\bm{v}_j)}\frac{[\bm{D}_0]_{rr}}{[\bm{D}_0]_{jj}}[(\bm{A}_0\bm{D}^{-1}_0)^{l}]_{kr}[\bm{A}_0\bm{D}^{-1}_0]_{jr}$. 
To prove the time complexity guarantee, at each iteration $\bm{v}_k$ computes at most $\alpha$ values, where $\alpha  = \min\left\{n, \frac{\left(d^{R+1}_{\text{max}} - 1\right)}{\left(d_{\text{max}} - 1\right)}\right\}$ is the upper bound on the size of the R-hop neighborhood $\forall \bm{v}\in \mathbb{V}$.  Each such computation requires at most $\mathcal{O}(d_{max})$ operations. Thus, the overall time complexity is given by $\mathcal{O}(\alpha Rd_{max})$.
\end{proof}

We are now ready to provide the proof of Lemma \ref{r_hop_rude_lemma}. 
\begin{proof}
From \textbf{Parts Two} and \textbf{Three} of Algorithm~\ref{Algo:DistRHop}, it is clear that node $\bm{v}_k$ computes $[\bm{b}_1]_k,[\bm{b}_2]_k,\ldots, [\bm{b}_d]_k$ and $[\bm{x}_d]_k, [\bm{x}_{d-1}]_k,\ldots, [\bm{x}_0]_k$, respectively. These are determined using the inverse approximated chain as follows
\begin{align}\label{eq_4}
\bm{b}_i &= (\bm{I} + (\bm{A}_{i-1}\bm{D}^{-1}_{i-1})\bm{b}_{i-1}  =\bm{b}_{i-1} + (\bm{A}_0\bm{D}^{-1}_0)^{2^{i-1}}\bm{b}_{i-1} \\ \nonumber
\bm{x}_{i} &= \frac{1}{2}[\bm{D}^{-1}_i\bm{b}_i + (\bm{I} +\bm{D}^{-1}_i\bm{A}_i)x_{i+1}] =\frac{1}{2}[\bm{D}^{-1}_0\bm{b}_i + \bm{x}_{i+1} + (\bm{D}^{-1}_0\bm{A}_0)^{2^i}\bm{x}_{i+1}]
\end{align}

Considering the computation of $[\bm{b}_1]_k,\ldots, [\bm{b}_d]_k$ for $\rho> i - 1$, we have
\begin{align*}
[\bm{b}_{i}]_k &= [\bm{b}_{i-1}]_k + [(\bm{A}_0\bm{D}^{-1}_0)^{2^{i-1}}\bm{b}_{i-1}]_k = [\bm{b}_{i-1}]_k + [\underbrace{\bm{A}_0\bm{D}^{-1}_0 \ldots \bm{A}_0\bm{D}^{-1}_0}_{2^{i-1}}\bm{b}_{i-1}]_k = [\bm{b}_{i-1}]_k + [\underbrace{\bm{A}_0\bm{D}^{-1}_0 \ldots \bm{A}_0\bm{D}^{-1}_0}_{2^{i-1}-1}\bm{u}^{(i-1)}_1]_k \dots \\
&= [\bm{b}_{i-1}]_k + \left[\bm{u}^{(i-1)}_{2^{i-1}}\right]_k, \ \ \ \text{with $\bm{u}^{(i-1)}_{j+1} = \bm{A}_0\bm{D}^{-1}_0\bm{u}^{(i-1)}_{j}$ for $j = 1,\ldots 2^{i-1}-1$.}
\end{align*}

Since $\bm{A}_0\bm{D}^{-1}_0$ has a sparsity pattern corresponding to 1-hop neighborhood (see Claim~\ref{claim_1}),  node $\bm{v}_k$ computes $\left[\bm{u}^{(i-1)}_{j+1}\right]_k$, based on $\bm{u}^{(i-1)}_{j}$, acquired from its 1-Hop neighbors. It is easy to see that $ \forall i \ \text{such that} \ i - 1 < \rho$ the computation of $[\bm{b}_i]_k$ requires $\mathcal{O}\left(2^{i-1}d_{max}\right)$ time steps. Thus, the computation of $[\bm{b}_1]_k, \ldots, [\bm{b}_{\rho}]_k$ requires $\mathcal{O}(2^{\rho}d_{\text{max}}) = \mathcal{O}(Rd_{\text{max}})$. Now, consider the computation of $[\bm{b}_i]_k$ but for $i - 1 \geq \rho$ 
\begin{align*}
[\bm{b}_{i}]_k &= [\bm{b}_{i-1}]_k + [(\bm{A}_0\bm{D}^{-1}_0)^{2^{i-1}}\bm{b}_{i-1}]_k  & =[\bm{b}_{i-1}]_k + [\underbrace{\bm{C}_0\ldots \bm{C}_0}_{l_{i-1}}\bm{b}_{i-1}]_k =[\bm{b}_{i-1}]_k + [\underbrace{\bm{C}_0 \ldots \bm{C}_0}_{l_{i-1}-1}\bm{u}^{(i-1)}_1]_k  =[\bm{b}_{i-1}]_k + \left[\bm{u}^{(i-1)}_{l_{i-1}}\right]_k
\end{align*}
with $\bm{C}_0 = (\bm{A}_0\bm{D}^{-1}_0)^R$,  $l_{i-1} = \frac{2^{i-1}}{R}$, and $\bm{u}^{(i-1)}_{j+1} = \bm{C}_0\bm{u}^{(i-1)}_{j}$ for $j=1,\ldots, l_{i-1} - 1$.
Since $\bm{C}_0$ has a sparsity pattern corresponding to R-hop neighborhood (see Claim~\ref{claim_1}), node $\bm{v}_k$ computes $[\bm{u}^{(i-1)}_{j+1}]_k$ based on the components of $\bm{u}^{(i-1)}_{j}$ attained from its R-hop neighbors. For each $i$ such that $i - 1 \geq \rho $ the computing $[\bm{b}_i]_k$ requires $\mathcal{O}\left(\frac{2^{i-1}}{R}\alpha\right)$ time steps, where $\alpha = \min\left\{n, \frac{\left(d^{R+1}_{\text{max}} - 1\right)}{\left(d_{\text{max}} - 1\right)}\right\}$ being the upper bound on the number of nodes in the $R-$ hop neighborhood $\forall \ \bm{v}\in \mathbb{V}$. Therefore, the  overall computation of $[\bm{b}_{\rho + 1}]_k, [\bm{b}_{\rho + 2}]_k, \ldots, [\bm{b}_d]_k$ is achieved in$\mathcal{O}\left(\frac{2^d}{R}\alpha\right)$ time steps. Finally, the time complexity for the computation of all of the values $[\bm{b}_1]_k,[\bm{b}_2]_k,\ldots, [\bm{b}_d]_k$ is $\mathcal{O}\left(\frac{2^d}{R}\alpha + Rd_{max}\right)$. Similar analysis can be applied to determine the computational complexity of $[\bm{x}_d]_k,[\bm{x}_{d-1}]_k,\ldots, [\bm{x}_1]_k$, i.e., \textbf{Part Three} of Algorithm~\ref{Algo:DistRHop}. Using Lemma~\ref{Rude_Alg_guarantee_Lemma}, we arrive at
$\bm{Z}^{\prime}_0\approx_{\epsilon_d} \bm{M}^{-1}_0$. 
Finally, using Claim \ref{claim_2}, the time complexity of $\text{RDistRSolve}$ (Algorithm~\ref{Algo:DistRHop}) is $\mathcal{O}\left(\frac{2^d}{R}\alpha + \alpha Rd_{max}\right)$.
\end{proof} 

\subsection{``Exact'' Distributed R-Hop SDDM Solver}
Having developed an R-hop version which computes a ``rude'' approximation to the solution of $\bm{M}_{0}\bm{x}=\bm{b}_{0}$, now we provide an exact R-hop solver presented in Algorithm~\ref{Alg_ExactRHop}. Similar to $\text{RDistRSolve}$, each node $\bm{v}_k$ receives the $k^{th}$ row $\bm{M}_0$, $[\bm{b}_0]_k$, $d$, $R$, and a precision parameter $\epsilon$ as inputs, and outputs the $k^{th}$ component of the $\epsilon$ close approximation of vector $\bm{x}^{\star}$. 

\begin{algorithm}
  \caption{ \hspace{-.5em} $\text{EDistRSolve}\left(\{[\bm{M}_0]_{k1},\ldots, [\bm{M}_0]_{kn}\}, [\bm{b}_0]_k, d, R,\epsilon\right)$}
  \label{Alg_ExactRHop}
  \begin{algorithmic}
\State \textbf{Initialize}: $[\bm{y}_0]_k = 0$, and $[\bm{\chi}]_k = \text{RDistRSolve}(\{[M_0]_{k1},\ldots, [M_0]_{kn}\}, [b_0]_k, d, R)$	
	\For {$t=1$ to $q$}
		\State  \hspace{-2em} $[\bm{u}^{(1)}_{t}]_k = [\bm{D}_0]_{kk}[\bm{y}_{t-1}]_k - \sum_{j: \bm{v}_j\in \mathbb{N}_{1}(\bm{v}_k)}[\bm{A}_{0}]_{kj}[\bm{y}_{t-1}]_j$
		 \State \hspace{-2em} $[\bm{u}^{(2)}_{t}]_k = \text{RDistRSolve}(\{[\bm{M}_0]_{k1},\ldots, [\bm{M}_0]_{kn}\}, [\bm{u}^{(1)}_{t}]_k, d, R)$
		\State  \hspace{-2em} $[\bm{y}_t]_k = [\bm{y}_{t-1}]_k - [\bm{u}^{(2)}_{t}]_k + [\bm{\chi}]_k$ 
	\EndFor \textbf{end for}	
    \State \textbf{return} $[\tilde{\bm{x}}]_k = [\bm{y}_q]_k$
  \end{algorithmic}
\end{algorithm}

\textbf{Analysis of Algorithm~\ref{Alg_ExactRHop}:} The following Lemma shows that $\text{EDistRSolve}$ computes the $k^{th}$ component of the $\epsilon$ close approximation to $\bm{x}^{\star}$ and provides the time complexity analysis. 

\begin{lemma}\label{Dist_Exact_algorithm_guarantee_lemma}
Let $\bm{M}_0 = \bm{D}_0 - \bm{A}_0$ be the standard splitting. Further, let $\epsilon_d < \sfrac{1}{3}\ln2$. Then Algorithm~\ref{Alg_ExactRHop} requires $\mathcal{O}\left(\log\frac{1}{\epsilon}\right)$ iterations to return the $k^{th}$ component of the $\epsilon$ close approximation to $\bm{x}^{\star}$.  
\end{lemma}
\begin{proof}
See Appendix.
\end{proof} 
Next, the following Lemma provides the time complexity analysis of $\text{EDistRSolve}$. 

\begin{lemma}\label{time_complexity_of_distresolve}
Let $\bm{M}_0 = \bm{D}_0 - \bm{A}_0$ be the standard splitting and let $\epsilon_d < \sfrac{1}{3}\ln 2 $, then $\text{EDistRSolve}$ requires 
$\mathcal{O}\left(\left(\sfrac{2^d}{R}\alpha + \alpha Rd_{max}\right)\log\left(\sfrac{1}{\epsilon}\right)\right)$ time steps. Moreover, for each node $\bm{v}_k$, $\text{EDistRSolve}$ only uses information from the R-hop neighbors.
\end{lemma}

\begin{proof}
See Appendix
\end{proof}
\subsection{Length of the Inverse Chain}
Again these introduced algorithms depend on the length of the inverse approximated chain, $d$. Here, we provide an analysis to determine the value of $d$ which guarantees $\epsilon_d < \frac{1}{3}\ln2$ in $\mathcal{C} = \{\bm{A}_0,\bm{D}_0,\bm{A}_1, \bm{D}_1,\ldots, \bm{A}_d, \bm{D}_d\}$. These results are summarized the following lemma

\begin{lemma}\label{eps_d_lemma}
Let $\bm{M}_0 = \bm{D}_0 - \bm{A}_0$ be the standard splitting and let $\kappa$ denote the condition number $\bm{M}_0$. Consider the inverse approximated chain $\mathcal{C} = \{\bm{A}_0,\bm{D}_0,\bm{A}_1, \bm{D}_1,\ldots, \bm{A}_d, \bm{D}_d\}$ with length $d =  \lceil \log \left(2\ln\left(\frac{\sqrt[3]{2}}{\sqrt[3]{2} - 1}\right)\kappa\right)\rceil$, then 
$\bm{D}_0\approx_{\epsilon_d} \bm{D}_0 - \bm{D}_0\left(\bm{D}^{-1}_0\bm{A}_0\right)^{2^d}$, 
with $\epsilon_d < \sfrac{1}{3}\ln2$.
\end{lemma}

\begin{proof}
The proof is similar to that of Section~\ref{Sec:Length} and can be found in the Appendix. 
\end{proof}

\section{Discussion \& Conclusions}
We developed a distributed version of the parallel SDDM solver of~\cite{Spielman} and proposed a fast decentralized solver for SDDM systems. Our approach is capable of acquiring $\epsilon$-close solutions for arbitrary $\epsilon$ in $\mathcal{O}\left(n^{3}\frac{\bm{\alpha}}{R}\frac{\bm{W}_{\text{max}}}{\bm{W}_{\text{min}}}\log\left(\frac{1}{\epsilon}\right)\right)$, with $n$ the number of nodes in graph $\mathbb{G}$, $\bm{W}_{\text{max}}$ and $\bm{W}_{\text{min}}$ denoting the largest and smaller weights of the edges in $\mathbb{G}$, respectively, $\bm{\alpha}=\min\left\{n,\frac{d_{\text{max}}^{R+1}-1}{d_{\text{max}}-1}\right\}$ representing the upper bound on the size of the R-Hop neighborhood $\forall \bm{v} \in \mathbb{V}$, and $\epsilon \in (0,\frac{1}{2}]$ as the precision parameter. After developing the full communication version, we proposed a generalization to the R-Hop case where communication is restricted. 

Our method is faster than state-of-the-art methods for iteratively solving linear systems. Typical linear methods, such as Jacobi iteration~\cite{Axelsson}, are guaranteed to converge if the matrix is \emph{strictly} diagonally dominant. We proposed a distributed algorithm that generalizes this setting, where it is guaranteed to converge in the SDD/SDDM scenario. Furthermore, the time complexity of linear techniques is $\mathcal{O}(n^{1+\beta} \log n)$, hence, a case of strictly diagonally dominant matrix $\bm{M}_{0}$ can be easily constructed to lead to a complexity of $\mathcal{O}(n^{4}\log n)$. Consequently, our approach not only generalizes the assumptions made by linear methods, but is also faster by a factor of $\log n$.

In centralized solvers, nonlinear methods (e.g., conjugate gradient descent~\cite{CG,Nocedal2006NO}, etc.) typically offer computational advantages over linear methods (e.g., Jacobi Iteration) for iteratively solving linear systems. These techniques, however, can not be easily decentralized. For instance, the stopping criteria for nonlinear methods require the computation of weighted norms of residuals (e.g., $||\bm{p}_{k}||_{\bm{M}_{0}}$ with $\bm{p}_{k}$ being the search direction at iteration $k$). To the best of our knowledge, the distributed computation of weighted norms is difficult. Namely using the approach in~\cite{Olshevsky}, this requires the calculation of the top singular value of $\bm{M}_{0}$ which amounts to a power iteration on $\bm{M}_{0}^{\mathsf{T}}\bm{M}_{0}$ leading to the loss of sparsity. Furthermore, conjugate gradient methods require global computations of inner products. 

Another existing method which we compare our results to is the recent work of the authors~\cite{Asuman} where a local and asynchronous solution for solving systems of linear equations is considered. In their work, the authors derive a complexity bound, for one component of the solution vector, of $\mathcal{O}\left(\min\left(d\epsilon^{\frac{\ln d}{\ln ||\bm{G}||_{2}}}, \frac{d n \ln \epsilon}{\ln ||\bm{G}||_{2}}\right)\right)$, with $\epsilon$ being the precision parameter, $d$ a constant bound on the maximal degree of $\mathbb{G}$, and $\bm{G}$ is defined as $\bm{x}=\bm{G}\bm{x}+\bm{z}$ which can be directly mapped to $\bm{A}\bm{x} = \bm{b}$. The relevant scenario to our work is when $\bm{A}$ is PSD and $\bm{G}$ is symmetric. Here, the bound on the number of multiplications is given by $\mathcal{O}\left(\min\left(d^{\frac{\kappa(\bm{A})+1}{2}\ln\frac{1}{\epsilon}}, \frac{\kappa(\bm{A})+1}{2}n d \ln\frac{1}{\epsilon}\right)\right)$, with $\kappa(\bm{A})$ being the condition number of $\bm{A}$. In the general case, when the degree depends on the number of nodes (i.e., $d = d(n)$), the minimum in the above bound will be the result of the second term ( $\frac{\kappa(\bm{A})+1}{2}n d \ln\frac{1}{\epsilon}$) leading to $\mathcal{O}\left(d(n) n \kappa(\bm{A})\ln \frac{1}{\epsilon}\right)$. Consequently, in such a general setting, our approach outperforms~\cite{Asuman} by a factor of $d(n)$. 


\newpage
%
\bibliographystyle{abbrv}
\bibliography{references}  

\begin{thebibliography}{10}

\bibitem{Reif199837}
Efficient approximate solution of sparse linear systems.
\newblock {\em Computers \& Mathematics with Applications}, 1998.

\bibitem{Axelsson}
O.~Axelsson.
\newblock {\em Iterative Solution Methods}.
\newblock Cambridge University Press, New York, NY, USA, 1994.

\bibitem{Batson}
J.~Batson, D.~A. Spielman, N.~Srivastava, and S.-H. Teng.
\newblock Spectral sparsification of graphs: Theory and algorithms.
\newblock {\em Commun. ACM}, 56(8):87--94, Aug. 2013.

\bibitem{Bertsekas:1989:PDC:59912}
D.~P. Bertsekas and J.~N. Tsitsiklis.
\newblock {\em Parallel and Distributed Computation: Numerical Methods}.
\newblock Prentice-Hall, Inc., Upper Saddle River, NJ, USA, 1989.

\bibitem{Blelloch}
G.~E. Blelloch, A.~Gupta, I.~Koutis, G.~L. Miller, R.~Peng, and K.~Tangwongsan.
\newblock Near linear-work parallel {SDD} solvers, low-diameter decomposition,
  and low-stretch subgraphs.
\newblock {\em CoRR}, abs/1111.1750, 2011.

\bibitem{citeulike}
E.~G. Boman and B.~Hendrickson.
\newblock {Support Theory for Preconditioning}.
\newblock {\em SIAM J. Matrix Anal. Appl.}, 25(3):694--717, 2003.

\bibitem{ExamplesOne}
E.~G. Boman, B.~Hendrickson, and S.~A. Vavasis.
\newblock Solving elliptic finite element systems in near-linear time with
  support preconditioners.
\newblock {\em CoRR}, cs.NA/0407022, 2004.

\bibitem{ExamplesThree}
P.~Christiano, J.~A. Kelner, A.~Madry, D.~A. Spielman, and S.-H. Teng.
\newblock Electrical flows, laplacian systems, and faster approximation of
  maximum flow in undirected graphs.
\newblock In {\em Proceedings of the Forty-third Annual ACM Symposium on Theory
  of Computing}, STOC '11, pages 273--282, New York, NY, USA, 2011. ACM.

\bibitem{ExamplesTwo}
S.~I. Daitch and D.~A. Spielman.
\newblock Faster approximate lossy generalized flow via interior point
  algorithms.
\newblock {\em CoRR}, abs/0803.0988, 2008.

\bibitem{joshi}
A.~Joshi.
\newblock Topics in optimization and sparse linear systems.
\newblock Technical report, Champaign, IL, USA, 1996.

\bibitem{CG}
E.~Kaasschieter.
\newblock Preconditioned conjugate gradients for solving singular systems.
\newblock {\em Journal of Computational and Applied Mathematics}, 24(1–2):265
  -- 275, 1988.

\bibitem{Kelner}
J.~A. Kelner and A.~Madry.
\newblock Faster generation of random spanning trees.
\newblock {\em CoRR}, abs/0908.1448, 2009.

\bibitem{KoutisMiller07}
I.~Koutis and G.~L. Miller.
\newblock A linear work, $o(n^{1/6})$ time, parallel algorithm for solving
  planar laplacians.
\newblock pages 1002--1011, January 2007.

\bibitem{Koutis}
I.~Koutis, G.~L. Miller, and R.~Peng.
\newblock Approaching optimality for solving {SDD} systems.
\newblock {\em CoRR}, abs/1003.2958, 2010.

\bibitem{Koutis2}
I.~Koutis, G.~L. Miller, and R.~Peng.
\newblock Solving {SDD} linear systems in time $\tilde{O}(m \log n
  \log(1/\epsilon))$.
\newblock {\em CoRR}, abs/1102.4842, 2011.

\bibitem{Asuman}
C.~E. Lee, A.~E. Ozdaglar, and D.~Shah.
\newblock Solving systems of linear equations: Locally and asynchronously.
\newblock {\em CoRR}, abs/1411.2647, 2014.

\bibitem{liu2013asynchronous}
J.~Liu, S.~Mou, and A.~S. Morse.
\newblock An asynchronous distributed algorithm for solving a linear algebraic
  equation.
\newblock In {\em Decision and Control (CDC), 2013 IEEE 52nd Annual Conference
  on}, pages 5409--5414. IEEE, 2013.

\bibitem{Nocedal2006NO}
J.~Nocedal and S.~J. Wright.
\newblock {\em Numerical Optimization}.
\newblock Springer, New York, 2nd edition, 2006.

\bibitem{Olshevsky}
A.~{Olshevsky}.
\newblock {Linear Time Average Consensus on Fixed Graphs and Implications for
  Decentralized Optimization and Multi-Agent Control}.
\newblock {\em ArXiv e-prints}, Nov. 2014.

\bibitem{Spielman}
R.~Peng and D.~A. Spielman.
\newblock An efficient parallel solver for {SDD} linear systems.
\newblock {\em CoRR}, abs/1311.3286, 2013.

\bibitem{DaanS}
D.~A. Spielman and S.~Teng.
\newblock Nearly-linear time algorithms for preconditioning and solving
  symmetric, diagonally dominant linear systems.
\newblock {\em CoRR}, abs/cs/0607105, 2006.

\bibitem{Spiel}
D.~A. Spielman and S.~Teng.
\newblock Spectral sparsification of graphs.
\newblock {\em CoRR}, abs/0808.4134, 2008.

\bibitem{Zhu03semi-supervisedlearning}
X.~Zhu, Z.~Ghahramani, and J.~Lafferty.
\newblock Semi-supervised learning using gaussian fields and harmonic
  functions.
\newblock In {\em IN ICML}, pages 912--919, 2003.

\end{thebibliography}
\newpage 

\section*{Appendix}
In this appendix, we will provide proof of the Lemmas in the original submission. 

\begin{mylemma}\label{lemma_approx_matrix_inverse}
Let $\bm{Z}_0\approx_{\epsilon}\bm{M}^{-1}_0$, and $\tilde{\bm{x}} = \bm{Z}_0\bm{b}_0$. Then $\tilde{\bm{x}}$ is $\sqrt{2^{\epsilon}(e^{\epsilon} - 1)}$ approximate solution of $\bm{M}_{0}\bm{x}=\bm{b}_{0}$.
\end{mylemma}
\begin{proof}
Let $\bm{x}^{\star}\in \mathbb{R}^{n}$ be the solution of $\bm{M}_0\bm{x}=\bm{b}_{0}$, then
\begin{align}\label{app_solut_for_approx_inv}
&||\bm{x}^{\star} - \tilde{\bm{x}}||^2_{\bm{M}_0} =(\bm{x}^{\star} - \tilde{\bm{x}})^T\bm{M}_0(\bm{x}^{\star} - \tilde{\bm{x}}) = (\bm{x}^{\star})^T\bm{M}_0\bm{x}^{\star} + (\tilde{\bm{x}})^T\bm{M}_0\tilde{\bm{x}} - 2(\bm{x}^{\star})^T\bm{M}_0\tilde{\bm{x}}
\end{align}
Consider each term separately in (\ref{app_solut_for_approx_inv}):
\begin{enumerate}
\item $(\bm{x}^{\star})^T\bm{M}_0\tilde{\bm{x}} = \bm{b}^{T}_0\bm{M}^{-1}_0\bm{M}_0\bm{Z}_0\bm{b}_0 = \bm{b}^T_0\bm{Z}_0\bm{b}_0$

\item $(\bm{x}^{\star})^T\bm{M}_0\bm{x}^{\star} = \bm{b}^T_0\bm{M}^{-1}_0\bm{M}_0\bm{M}^{-1}_0\bm{b}_0 = \bm{b}^T_0\bm{M}^{-1}_0\bm{b}_0 \le e^{\epsilon}\bm{b}^T_0\bm{Z}_0\bm{b}_0$

\item $\tilde{\bm{x}}^T\bm{M}_0\tilde{\bm{x}} = \bm{b}^T_0\bm{Z}_0\bm{M}_0\bm{Z}_0\bm{b}_0 \le e^{\epsilon}\bm{b}^T_0\bm{Z}_0\bm{b}_0$

where in the last step we used that if $\bm{Z}_0\approx_{\epsilon}\bm{M}^{-1}_0$, then $\bm{M}_0 \approx_{\epsilon}\bm{Z}^{-1}_0$.
\end{enumerate}
Therefore, (\ref{app_solut_for_approx_inv}) can be rewritten as:
\begin{equation}\label{inter_1}
||\bm{x}^{\star} - \tilde{\bm{x}}||^2_{\bm{M}_0} \le 2(e^{\epsilon} - 1)\bm{b}^T_0\bm{Z}_0\bm{b}_0
\end{equation}
Combining (\ref{inter_1}) with $\bm{b}^T_0\bm{Z}_0\bm{b}_0 = (\bm{x}^{\star})^T\bm{M}_0\bm{Z}_0\bm{M}_0\bm{x}^{\star}  \le e^{\epsilon}(\bm{x}^{\star})^T\bm{M}_0\bm{x}^{\star}$:
\begin{align*}
&||\bm{x}^{\star} - \tilde{\bm{x}}||^2_{\bm{M}_0} \le 2(e^{\epsilon} - 1)e^{\epsilon}(\bm{x}^{\star})^T\bm{M}_0\bm{x}^{\star} = 2(e^{\epsilon} - 1)e^{\epsilon}||\bm{x}^{\star}||^2_{\bm{M}_0}
\end{align*}
\end{proof}

\begin{mylemma}\label{Rude_Dec_Alg_guarantee_Lemma}
Let $\bm{M}_0 = \bm{D}_0 - \bm{A}_0$ be the standard splitting of $\bm{M}_{0}$. Let $\bm{Z}^{\prime}_0$ be the operator defined by $\text{DistrRSolve}([\{[\bm{M}_0]_{k1},\ldots, [\bm{M}_0]_{kn}\}, [\bm{b}_0]_k, d)$ (i.e., $\bm{x}_0 = \bm{Z}^{\prime}_0\bm{b}_0$). Then
\begin{equation*}
\bm{Z}^{\prime}_0\approx_{\epsilon_d} \bm{M}^{-1}_0
\end{equation*}
Moreover, Algorithm~\ref{Algo:DisRudeApprox} requires  $\mathcal{O}\left(dn^2\right)$ time steps. 
\end{mylemma}
\begin{proof} 
The proof commences by showing that $\left(\bm{D}^{-1}_0\bm{A}_0\right)^r$ and $\left(\bm{A}_0\bm{D}^{-1}_0\right)^{-r}$ have a sparsity pattern corresponding to the r-hop neighborhood for any $r\in \mathbb{N}$. This case be shown using induction as follows
\begin{enumerate}
\item \text{Base case:} If $r = 1$, we have
\begin{equation*}
[\bm{A}_0\bm{D}^{-1}_0]_{ij} = \begin{cases} \frac{[\bm{A}_0]_{ij}}{[\bm{D}_0]_{ii}} &\mbox{if } j: \bm{v}_j\in \mathbb{N}_{1}(\bm{v}_i), \\ 
0 & \mbox{otherwise }. \end{cases}
\end{equation*}
Therefore, $\bm{A}_0\bm{D}^{-1}_0$ has sparsity pattern corresponding to the 1-hop neighborhood.
\end{enumerate}
Assume that for all $1\le p\le r-1$, $(\bm{A}_0\bm{D}^{-1}_0)^p$ has a sparsity patter corresponding to the p-hop neighborhood.
\begin{enumerate}
\item[2.] Now, consider $(\bm{A}_0\bm{D}^{-1}_0)^{r}$, where
\begin{equation}\label{sparisty_lemma_2}
[(\bm{A}_0\bm{D}^{-1}_0)^{r}]_{ij} = \sum_{k=1}^{n}[(\bm{A}_0\bm{D}^{-1}_0)^{r-1}]_{ik}[\bm{A}_0\bm{D}^{-1}_0]_{kj}
\end{equation}
Since $\bm{A}_0\bm{D}^{-1}_0$ is non negative then it is easy to see that $[(\bm{A}_0\bm{D}^{-1}_0)^{r}]_{ij}\ne 0$ if and only if there exists $k$ such that $\bm{v}_k\in \mathbb{N}_{r-1}(\bm{v}_i)$ and $\bm{v}_k\in \mathbb{N}_1(\bm{v}_j)$ (i.e., $\bm{v}_j\in \mathbb{N}_{r}(\bm{v}_i)$).
\end{enumerate}
For $\bm{D}^{-1}_0\bm{A}_0$, the same results can be derived similarly.

Please notice that in \textbf{Part One} of $\text{DistrRSolve}$ algorithm node $\bm{v}_k$ computes (in a distributed fashion) the components $[\bm{b}_1]_k$ to $[\bm{b}_d]_{k}$ using the inverse approximated chain $\mathcal{C} = \{\bm{A}_0,\bm{D}_0,\bm{A}_1, \bm{D}_1,\ldots, \bm{A}_d, \bm{D}_d\}$. Formally, 
\begin{align*}
\bm{b}_{i} &= \left[\bm{I} + \left(\bm{A}_0\bm{D}^{-1}_0\right)^{2^{i-1}}\right]\bm{b}_{i-1} = \bm{b}_{i-1} + \left(\bm{A}_0\bm{D}^{-1}_0\right)^{2^{i-2}} \\
&\hspace{18em}\left(\bm{A}_0\bm{D}^{-1}_0\right)^{2^{i-2}}\bm{b}_{i-1}
\end{align*}
Clearly, at the $i^{th}$ iteration node $\bm{v}_k$ requires the $k^{th}$ row of $\left(\bm{A}_{0}\bm{D}^{-1}_{0}\right)^{2^{i-2}}$ (i.e., the $k^{th}$ row from the previous iteration) in addition to the $j^{th}$ row of $\left(\bm{A}_0\bm{D}^{-1}_0\right)^{2^{i-2}}$ from all nodes $\bm{v}_j\in \mathbb{N}_{2^{i-1}}\left(\bm{v}_k\right)$ to compute the $k^{th}$ row of $\left(\bm{A}_{0}\bm{D}^{-1}_{0}\right)^{2^{i-1}}$. 

For computing $\left[\left(\bm{A}_0\bm{D}^{-1}_0\right)^{2^{i-1}}\right]_{kj}$, node $\bm{v}_k$ requires the $k^{th}$ row and $j^{th}$ column of $\left(\bm{A}_0\bm{D}^{-1}_0\right)^{2^{i-2}}$. The problem, however, is that node $\bm{v}_j$ can only send the $j^{th}$ row of $\left(\bm{A}_0\bm{D}^{-1}_0\right)^{2^{i-2}}$ which can be easily seen not to be see that symmetric. To overcome this issue, node $\bm{v}_k$ has to compute the $j^{th}$ column of $\left(\bm{A}_0\bm{D}^{-1}_0\right)^{2^{i-2}}$ based on its $j^{th}$ row.  The fact that $\bm{D}^{-1}_0\left(\bm{A}_0\bm{D}^{-1}_0\right)^{2^{i-2}}$ is symmetric, manifests that for $r = 1,\ldots, n$

\begin{equation*}
\frac{\left[\left(\bm{A}_0\bm{D}^{-1}_0\right)^{2^{i-2}}\right]_{rj}}{[\bm{D}_0]_{rr}} = \frac{\left[\left(\bm{A}_0\bm{D}^{-1}_0\right)^{2^{i-2}}\right]_{jr}}{[\bm{D}_0]_jj}
\end{equation*}
Hence, for all $r=1,\ldots,n$
\begin{equation}\label{columns_from_rows}
\left[\left(\bm{A}_0\bm{D}^{-1}_0\right)^{2^{i-2}}\right]_{rj} = \frac{[\bm{D}_0]_{rr}}{[\bm{D}_0]_{jj}}\left[\left(\bm{A}_0\bm{D}^{-1}_0\right)^{2^{i-2}}\right]_{jr}
\end{equation}     
Now, lets analyze the time complexity of computing components ${[b_1]_k, [b_2]_k,\ldots, [b_d]_k}$.\\\newline
 \textbf{Time Complexity Analysis:} At each iteration $i$, node $\bm{v}_k$ receives the $j^{th}$ row of $\left(\bm{A}_0\bm{D}^{-1}_0\right)^{2^{i-2}}$ from all nodes $\bm{v}_j\in \mathbb{N}_{2^{i-1}}(\bm{v}_k)$. using Equation~\ref{columns_from_rows}, node $\bm{v}_{k}$ computes the corresponding columns as well as the product of these columns with the $k^{th}$ row of $\left(\bm{A}_0\bm{D}^{-1}_0\right)^{2^{i-2}}$. Therefore, the time complexity at the $i^{th}$ iteration is $\mathcal{O}\left(n^2 + \text{diam}\left(\mathbb{G}\right)\right)$, where $n^2$ is responsible for the $k^{th}$ row computation, and $\text{diam}\left(\mathbb{G}\right)$ represents the communication cost between the nodes.  Using the fact that $\text{diam}\left(\mathbb{G}\right)\le n$, the total complexity of $\textbf{Part One}$ in $\text{DistrRSolve}$ algorithm is $\mathcal{O}\left(dn^2\right)$.

In \textbf{Part Two}, node $\bm{v}_k$ computes (in a distributed fashion) ${[\tilde{\bm{x}}_{d-1}]_k, [\tilde{\bm{x}}_{d-2}]_k,\ldots, [\tilde{\bm{x}}_{0}]_k}$ using the same inverse approximated chain $\mathcal{C} = \{\bm{A}_0,\bm{D}_0,\bm{A}_1, \bm{D}_1,\ldots, \bm{A}_d, \bm{D}_d\}$. 
\begin{align}
\bm{x}_{i} = \frac{1}{2}\bm{D}^{-1}_0\bm{b}_{i} &+ \frac{1}{2}\left[\bm{I} + (\bm{D}^{-1}_0\bm{A}_0)^{2^i}\right]\bm{x}_{i+1} = \frac{1}{2}\bm{D}^{-1}_0\bm{b}_{i} + \frac{1}{2}\bm{x}_{i+1}  \\\nonumber
&\hspace{3em}+\frac{1}{2}\left(\bm{D}^{-1}_0\bm{A}_0\right)^{2^{i-1}}\left(\bm{D}^{-1}_0\bm{A}_0\right)^{2^{i-1}}\bm{x}_{i+1}
\end{align}
for $i=d-1, \ldots, 1$. Thus,
\begin{equation*}
\bm{x}_{0} = \frac{1}{2}\bm{D}^{-1}_0\bm{b}_{0} + \frac{\bm{x}_1}{2} + \frac{1}{2}\left(\bm{D}^{-1}_0\bm{A}_0\right)\bm{x}_{1}
\end{equation*}
Similar to the analysis of \textbf{Part One} of $\text{DistrRSolve}$ algorithm the time complexity of \textbf{Part Two} as well as the time complexity of the whole algorithm is $\mathcal{O}\left(dn^2\right)$.

Finally, using Lemma~\ref{Rude_Alg_guarantee_Lemma} of the original paper for the inverse approximated chains $\mathcal{C} = \{\bm{A}_0,\bm{D}_0, \bm{A}_1, \bm{D}_1,\ldots, \bm{A}_d, \bm{D}_d\}$ yields:

\begin{equation*}
\bm{Z}^{\prime}_0\approx_{\epsilon_d} \bm{M}^{-1}_0.
\end{equation*}
\end{proof}

\begin{mylemma}\label{Dist_Exact_algorithm_guarantee_lemma_1}
Let $\bm{M}_0 = \bm{D}_0 - \bm{A}_0$ be the standard splitting. Further, let $\epsilon_d < \frac{1}{3}\ln2$ in the nverse approximated chain $\mathcal{C} = \{\bm{A}_0,\bm{D}_0,\bm{A}_1, \bm{D}_1,\ldots, \bm{A}_d, \bm{D}_d\}$. Then $\text{DistrESolve}(\{[\bm{M}_0]_{k1},\ldots, [\bm{M}_0]_{kn}\}, [b_0]_k, d, \epsilon)$ requires $\mathcal{O}\left(\log\frac{1}{\epsilon}\right)$ iterations to return the $k^{th}$ component of the $\epsilon$ close approximation for $x^{\star}$.
\end{mylemma}
\begin{proof}
Notice that iterations in $\text{DistrESolve}$ corresponds to Preconditioned Richardson Iteration:
\begin{equation*}
\bm{y}_t = \left[\bm{I} - \bm{Z}'_0\bm{M}_0\right]\bm{y}_{t-1} + \bm{Z}_0\bm{b}_0
\end{equation*}
where $\bm{Z}'_0$ is the operator defined by $\text{DistrRSolve}$ and $\bm{y}_0 = \bm{0}$. Therefore, from Lemma \ref{Rude_Dec_Alg_guarantee_Lemma}:
\begin{equation*}
\bm{Z}'_0 \approx_{\epsilon_d} \bm{M}^{-1}_0 
\end{equation*}
Finally, applying Lemma \ref{Exact_Alg_guarantee_lemma} of the main submission gives that $\text{DistrESolve}$ algorithm needs $\mathcal{O}\left(\log\frac{1}{\epsilon}\right)$ iterations to $k^{th}$ component of the $\epsilon$ approximated solution for $x^{\star}$.
\end{proof}

\begin{mylemma}\label{time_complexity_of_distresolve}
Let $\bm{M}_0 =\bm{D}_0 - \bm{A}_0$ be the standard splitting. Further, let $\epsilon_d < \frac{1}{3}\ln2$ in the inverse approximated chain $\mathcal{C} = \{\bm{A}_0,\bm{D}_0,\bm{A}_1, \bm{D}_1,\ldots, \bm{A}_d, \bm{D}_d\}$. Then, $\text{DistrESolve}(\{[\bm{M}_0]_{k1},\ldots, [\bm{M}_0]_{kn}\}, [\bm{b}_0]_k, d, \epsilon)$ requires $\mathcal{O}\left(dn^2\log(\frac{1}{\epsilon})\right)$ time steps.
\end{mylemma}

\begin{proof}
Each iteration of $\text{DistrESolve}$ algorithm calls $\text{DistRSolve}$ routine, therefore, using Lemmas \ref{Rude_Dec_Alg_guarantee_Lemma} and \ref{Dist_Exact_algorithm_guarantee_lemma_1} the total time complexity of f $\text{DistrESolve}$ algorithm is $\mathcal{O}\left(dn^2\log(\frac{1}{\epsilon})\right)$ time steps

\end{proof}

\begin{claim}
Let $\kappa$ be the condition number of $\bm{M}_0 = \bm{D}_0 - \bm{A}_0$, and $\{\lambda_i\}^{n}_{i=1}$ denote the eigenvalues of $\bm{D}^{-1}_0\bm{A}_0$. Then,
$|\lambda_i| \le 1 - \frac{1}{\kappa}$,  
for all $i=1,\ldots, n$
\end{claim}
\begin{proof}
See Proposition $5.3$ in ~\cite{Spielman}. 
\end{proof}

\begin{claim}
Let $\bm{M}$ be an SDDM matrix and consider the splitting $\bm{M} =\bm{D} -\bm{A}$, with $\bm{D}$ being non negative diagonal and $\bm{A}$ being symmetric non negative. Further, assume that the eigenvalues of $\bm{D}^{-1}\bm{A}$ lie between $-\alpha$ and $\beta$. Then, $
(1 - \beta)\bm{D} \preceq \bm{D} -\bm{A} \preceq (1 + \alpha)\bm{D}$. 
\end{claim}
\begin{proof} 
See Proposition $5.4$ in ~\cite{Spielman}.
\end{proof}

\begin{mylemma}\label{Dist_Exact_algorithm_guarantee_lemma}
Let $\bm{M}_0 = \bm{D}_0 - \bm{A}_0$ be the standard splitting. Further, let $\epsilon_d < \sfrac{1}{3}\ln2$. Then Algorithm~\ref{Alg_ExactRHop} requires $\mathcal{O}\left(\log\frac{1}{\epsilon}\right)$ iterations to return the $k^{th}$ component of the $\epsilon$ close approximation to $\bm{x}^{\star}$.  
\end{mylemma}
\begin{proof}
Please note that the iterations of $\text{EDistRSolve}$ correspond to a distributed version of the  preconditioned Richardson iteration scheme
\begin{equation*}
\bm{y}_{t} = [\bm{I} - \bm{Z}^{\prime}_0\bm{M}_0]\bm{y}_{t-1} + \bm{Z}^{\prime}_0\bm{b}_0
\end{equation*}
with $\bm{y}_0 = 0$ and $\bm{Z}^{\prime}_0$ being the operator defined by $\text{RDistRSolve}$. From Lemma~\ref{r_hop_rude_lemma} it is clear that
$\bm{Z}'_0 \approx_{\epsilon_d}\bm{M}^{-1}_0$. 
Applying Lemma~\ref{Exact_Alg_guarantee_lemma}, provides that $\text{EDistRSolve}$ requires $ \mathcal{O}\left(\log\sfrac{1}{\epsilon}\right)$ iterations to return the $k^{th}$ component of the $\epsilon$ close approximation to $\bm{x}^{\star}$. Finally, since $\text{EDistRSolve}$ uses procedure $\text{RDistRSolve}$ as a subroutine, it follows that for each node $\bm{v}_k$ only communication between the R-hope neighbors is allowed. 
\end{proof} 

\begin{mylemma}\label{time_complexity_of_distresolve}
Let $\bm{M}_0 = \bm{D}_0 - \bm{A}_0$ be the standard splitting and let $\epsilon_d < \sfrac{1}{3}\ln 2 $, then $\text{EDistRSolve}$ requires 
$\mathcal{O}\left(\left(\sfrac{2^d}{R}\alpha + \alpha Rd_{max}\right)\log\left(\sfrac{1}{\epsilon}\right)\right)$ time steps. Moreover, for each node $\bm{v}_k$, $\text{EDistRSolve}$ only uses information from the R-hop neighbors.
\end{mylemma}

\begin{proof}
Notice that at each iteration $\text{EDistRSolve}$ calls  $\text{RDistRSolve}$ as a subroutine, therefore, for each node $v_k$ only R-hop communication is allowed. Lemma \ref{r_hop_rude_lemma} gives that the time complexity of each iteration is $\mathcal{O}\left(\frac{2^d}{R}\alpha + \alpha Rd_{max}\right)$, and using Lemma~\ref{Dist_Exact_algorithm_guarantee_lemma} immediately gives that the time complexity of $\mathcal{O}\left(\left(\sfrac{2^d}{R}\alpha + \alpha Rd_{max}\right)\log\left(\sfrac{1}{\epsilon}\right)\right)$.
\end{proof}

\begin{mylemma}\label{eps_d_lemma_2}
Let $\bm{M}_0 = \bm{D}_0 - \bm{A}_0$ be the standard splitting and let $\kappa$ denote the condition number $\bm{M}_0$. Consider the inverse approximated chain $\mathcal{C} = \{\bm{A}_0,\bm{D}_0,\bm{A}_1, \bm{D}_1,\ldots, \bm{A}_d, \bm{D}_d\}$ with length $d =  \lceil \log \left(2\ln\left(\frac{\sqrt[3]{2}}{\sqrt[3]{2} - 1}\right)\kappa\right)\rceil$, then 
$\bm{D}_0\approx_{\epsilon_d} \bm{D}_0 - \bm{D}_0\left(\bm{D}^{-1}_0\bm{A}_0\right)^{2^d}$, 
with $\epsilon_d < \sfrac{1}{3}\ln2$.
\end{mylemma}

\begin{proof}
The proof takes similar steps as in Lemma \ref{eps_d_lemma}.
\end{proof}

\end{document}